\documentclass[journal=gmj]{CUP-JNL-DTM}%

\usepackage{graphicx}
\usepackage{multicol,multirow}
\usepackage{amsmath,amssymb,amsfonts}
\usepackage{mathrsfs}
\usepackage{amsthm}
\usepackage{rotating}
\usepackage{appendix}
\usepackage{ifpdf}
\usepackage[T1]{fontenc}
\usepackage{newtxtext}
\usepackage{newtxmath}
\usepackage{textcomp}
\usepackage{xcolor}
\usepackage{lipsum}
\usepackage[colorlinks,allcolors=blue]{hyperref}
\usepackage{bm}
\usepackage{algorithm}
\usepackage{algorithmic}
\usepackage{longtable}
\makeatletter
\long\def\LT@makecaption#1#2#3{%
  \LT@mcol\LT@cols c{%
    \hbox to\z@{%
      \hss
      \parbox[t]\LTcapwidth{%
        \@tablecaption{#2}{#3}%
        \vskip\belowcaptionskip
      }%
      \hss
    }%
  }%
}
\makeatother
\usepackage{pdflscape}
\usepackage{array}
\graphicspath{{./}{../figures/}}

\newtheorem{theorem}{Theorem}[section]

\newtheorem{remark}{Remark}[section]
\newtheorem{proposition}{Proposition}[section]

\theoremstyle{definition}
\numberwithin{equation}{section}

\jname{Data/Math}
\articletype{ARTICLE TYPE}
\jyear{YEAR}

\begin{document}

\begin{Frontmatter}

\title[SHIM]{Bayesian Variable Selection for High-Dimensional 
Predictors with Missing Psychometric Outcomes}

\author[1]{Zongyue Teng}
\author[2]{Shujie Ma}
\author[3,4,5]{Timothy J.\ Hohman} 
\author[3,4,5]{Angela L.\ Jefferson} 
\author[1,4,5]{Panpan Zhang}

\authormark{Teng \textit{et al}.}

\address[1]{\orgdiv{Department of Biostatistics}, 
\orgname{Vanderbilt University Medical Center}, 
\orgaddress{\city{Nashville}, \postcode{37203}, \state{TN}, 
\country{USA}}}

\address[2]{\orgdiv{Department of Statistics}, 
\orgname{University of California, Riverside}, 
\orgaddress{\city{Riverside}, \postcode{92521}, \state{CA}, 
\country{USA}}}

\address[3]{\orgdiv{Department of Neurology}, \orgname{Vanderbilt 
University Medical Center}, \orgaddress{\city{Nashville}, 
\postcode{37232}, \state{TN},  \country{USA}}}

\address[4]{\orgdiv{Vanderbilt Memory and Alzheimer's Center}, 
\orgname{Vanderbilt University Medical Center}, 
\orgaddress{\city{Nashville}, \postcode{37203}, \state{TN},  
\country{USA}}}

\address[5]{\orgdiv{Vanderbilt Alzheimer's Disease Research Center}, 
\orgname{Vanderbilt University Medical Center}, 
\orgaddress{\city{Nashville}, \postcode{37203}, \state{TN},  
\country{USA}}. \email{panpan.zhang@vumc.org}}
	
\authormark{Teng et al.}



\abstract{High-dimensional, multimodal predictors and partially 
observed multivariate outcomes are common in psychometric research. 
However, existing regularization methods often do not accommodate 
hierarchical predictor structures and are primarily designed for 
univariate outcomes. We propose SHIM, a Bayesian framework for 
structured variable selection that combines hierarchical horseshoe 
shrinkage with a Bayesian treatment of missing outcomes. The 
framework jointly accommodates predictor hierarchies, dependence 
among outcomes, and incomplete multivariate responses. We establish 
theoretical properties of the proposed prior specification and 
evaluate SHIM through simulation studies. The results demonstrate 
that SHIM balances sensitivity with false-positive control while 
yielding accurate coefficient estimates and well-calibrated 
uncertainty quantification. We further apply SHIM to data from an 
Alzheimer's disease cohort to characterize associations between 
multimodal neuroimaging measures and multivariate neuropsychological 
outcomes and to generate posterior-based multiple imputations for 
downstream analyses of the relationships between fluid biomarkers 
and cognition. An R package, \texttt{shim}, is publicly available to 
facilitate implementation.}

\end{Frontmatter}


\localtableofcontents

\section{Introduction}
\label{sec:intro}

Modern psychometric studies increasingly integrate multimodal 
data to characterize cognitive and behavioral functioning 
from diverse biological perspectives. In neurological and 
psychiatric research, for example, structural and functional 
neuroimaging, diffusion imaging, genomics, proteomics, and other 
molecular biomarkers have been collected alongside comprehensive 
neuropsychological assessments to investigate the complex mechanisms 
underlying cognitive impairment and disease progression 
\citep{fan2008spatial, moore2020lower, rokicki2021multimodal}. 
Compared with analyses based on a single data source, multimodal 
studies may provide a richer and more comprehensive characterization 
of biological systems by leveraging complementary information across 
data modalities, potentially improving the identification and 
prediction of biomarkers associated with clinically meaningful 
psychometric outcomes~\cite{kline2022multimodal, 
venugopalan2021multimodal}. 

Meanwhile, the integration of multimodal data substantially 
increases the statistical complexity of psychometric studies. First, 
each modality may contain hundreds or thousands of candidate 
predictors, leading to high-dimensional regression problems in which 
the number of predictors may be comparable to or exceed the sample 
size. Conventional regression methods are therefore prone to unstable
estimation, multicollinearity, and overfitting, inducing
regularization and variable 
selection~\cite{carvalho2010thehorseshoe, 
tibshirani1996regression}. Moreover, high-dimensional predictors 
collected from multimodal data often exhibit intrinsic hierarchical 
organization. For example, structural magnetic resonance imaging 
(MRI) measurements may be organized according to anatomical regions 
within imaging modalities, while molecular biomarkers may be grouped 
according to biological pathways or functional categories. Such 
hierarchical structures further induce complex correlation patterns 
among predictors and suggest that biologically related variables 
should be modeled jointly rather than treated as independent 
features. Existing regularization methods, however, typically ignore 
these multilevel relationships or accommodate only a single grouping 
structure, potentially limiting both statistical efficiency and 
scientific interpretability~\cite{alt2025hierarchical, 
xu2016bayesian, yuan2006model}.

An additional challenge arises from the multivariate nature of 
psychometric outcomes. Modern psychometric studies routinely assess 
multiple cognitive or behavioral domains, such as memory, executive 
function, language, and attention, which reflect distinct yet 
correlated aspects of the same underlying neuropsychological 
process~\cite{weintraub2018version}. Modeling these outcomes jointly 
allows statistical procedures 
to borrow information across domains through their dependence, 
potentially improving estimation efficiency and statistical power 
relative to separate univariate 
analyses~\cite{brown1998multivariate, kundu2021bayesian, 
zellner1962anefficient}. In addition, many biomarkers may exhibit 
effects that influence multiple psychometric outcomes 
simultaneously~\cite{jack2018nia, sattlecker2014alzheimer}. Joint 
modeling therefore provides a coherent 
framework for identifying predictors with shared or outcome-specific 
associations while appropriately accounting for the correlation 
structure among outcomes. Despite the availability of multivariate
extensions, many widely used high-dimensional variable-selection
methods were originally developed for single-outcome regression
settings and may fail to exploit dependence across multiple 
outcomes~\cite{fan2010aselective}.

Existing methods have addressed different aspects of the statistical
challenges arising from multimodal psychometric studies, although
typically in isolation. To accommodate structured predictors,
group-based regularization methods, including the group lasso and
Bayesian sparse group selection, have been developed to exploit 
grouping information among predictors~\cite{xu2015bayesian, 
yuan2006model}. More recently, hierarchical
Bayesian shrinkage priors have been proposed to accommodate nested
predictor structures through adaptive regularization at multiple
levels~\citep{alt2025hierarchical, xu2016bayesian}. While these
approaches improve variable selection by incorporating predictor 
structure, they have primarily focused on univariate responses and 
generally do not account for dependence among multiple psychometric 
outcomes. Multivariate Bayesian variable selection methods have also 
been developed to exploit correlations among multiple outcomes and 
identify predictors with shared effects across 
responses~\citep{brown1998multivariate, kundu2021bayesian, 
liquet2017bayesian}. However, these methods typically treat 
predictors as independent or accommodate a single grouping structure.
Finally, methods that have been devoted to handling incomplete
outcomes through multiple imputation and joint modeling approaches
are not designed for high-dimensional variable selection with
hierarchically structured predictors~\cite{rubin1987multiple, 
tanner1987calculation}. Therefore, there remains a need for a 
unified statistical framework that simultaneously accommodates 
multilevel predictor structures, multivariate psychometric outcomes, 
and incomplete outcome data.

To address these challenges, we propose SHIM (Shrinkage for
Hierarchical predictors and Incomplete Multivariate outcomes), a
Bayesian framework for variable selection with high-dimensional,
multimodal predictors and partially observed multivariate
psychometric outcomes. SHIM considers a hierarchical horseshoe prior 
for the regression coefficients to accommodate the multilevel
structure of multimodal predictors while jointly modeling
correlated psychometric outcomes through a multivariate regression
model. Missing outcomes are treated as latent variables and inferred
jointly with the model parameters within a Bayesian framework. In 
addition to direct posterior inference for the model parameters, 
posterior samples of the missing outcomes naturally yield
a posterior-based multiple imputation representation for
subsequent analyses under statistical models requiring completed
outcome data. An open-source R package, \texttt{shim}, is available 
at \url{https://github.com/zongyue-teng/shim} to facilitate the 
implementation of SHIM.

The remainder of the manuscript is organized as follows. 
Section~\ref{sec:method} introduces the SHIM framework, including 
the prior specification, posterior computation, posterior-based 
multiple-imputation strategy, pseudocode algorithm, and 
probabilistic justification of the prior specification. 
Section~\ref{sec:sim} presents a simulation study comprising eight 
primary scenarios that vary in predictor dimension, modality and 
group structure, and proportion of missing outcomes. This section 
also reports four sets of sensitivity analyses assessing the 
robustness of SHIM and an additional experiment evaluating its 
multiple-imputation performance. Section~\ref{sec:app} presents two 
applications to an Alzheimer's disease cohort, where the first 
examines associations between multimodal neuroimaging measures and 
multivariate outcomes, and the second evaluates multiple-imputation 
performance in a downstream analysis of the relationship between 
cerebrospinal fluid biomarkers and cognition. Finally, 
Section~\ref{sec:dis} summarizes the proposed method, discusses its 
limitations, and outlines directions for future research.

\section{Method}
\label{sec:method}

\subsection{Notations}
\label{sec:notation}

Consider a study with $n$ independent individuals indexed by
$i = 1, 2, \ldots, n$. For each individual, let $\bm{y}_i = (y_{i1}, 
y_{i2}, \ldots, y_{iQ})^\top \in \mathbb{R}^Q$ denote a vector of 
psychometric outcomes measured across $Q$ domains. In many clinical 
and behavioral studies, some components of $\bm{y}_i$ may be missing 
due to missed visits, inability to complete assessments, or study 
design~\cite{narhi2009treating, suen2026modeling}. Let $\bm{r}_i = 
(r_{i1}, r_{i2}, \ldots, r_{iQ})^\top$
denote the corresponding missingness indicators, with $r_{iq} = 1$ if
$y_{iq}$ is observed and $r_{iq} = 0$ otherwise. Accordingly, we are 
able to write $\bm{y}_i = (\bm{y}_{i, \mathrm{obs}}, 
\bm{y}_{i,\mathrm{mis}})$ to distinguish the observed and missing 
components.

For each individual $i$, a high-dimensional predictor vector
$\bm{x}_i = (x_{i1}, x_{i2}, \ldots, x_{iP})^\top \in \mathbb{R}^P$
is observed. Throughout the paper, predictors are assumed to be fully
observed and standardized to have mean zero and unit variance. The
feature dimension $P$ may be substantially larger than the sample size
$n$. In psychometric and biomedical applications, these predictors may
represent high-dimensional features derived from one or multiple data
modalities, such as regional brain imaging measurements, diffusion
tract characteristics, or other high-dimensional 
biomarkers~\cite{fan2008spatial, rokicki2021multimodal}. In addition 
to the high-dimensional predictors, let
$\bm{z}_i = (z_{i1}, z_{i2}, \ldots, z_{iL})^\top \in \mathbb{R}^{L}$
denote a vector of low-dimensional covariates, where $L$ is fixed
and typically much smaller than both $n$ and $P$. These covariates
may include demographic and clinical variables, such as age, sex,
education, and {\em apolipoprotein E} ({\em APOE}) genotype, which 
are conventionally adjusted for in psychometric and biomedical 
studies. Throughout the paper, these covariates are assumed to be 
fully observed and are not subject to hierarchical shrinkage.

Many high-dimensional predictors possess a natural hierarchical 
organization. For example, multimodal neuroimaging predictors may be 
organized according to imaging modality and anatomical region. To 
accommodate such structure, we assume that predictors are organized 
into three nested levels: (1) modality level indexed by $m = 1, 2, 
\ldots, M$, (2) group level within modality $m$ indexed by $g = 1, 2, 
\ldots, G_m$, and (3) feature level within group $g$ indexed by $k = 
1, 2, \ldots, P_{mg}$. Under this organization, the total number of 
predictors is
\[
P = \sum_{m = 1}^{M} \sum_{g = 1}^{G_m} P_{mg}.
\]
The proposed method in this work is developed under this hierarchical
structure, but the framework can be extended naturally to more complex
hierarchical organizations. For notation convenience, predictors are
subsequently re-indexed by $k = 1, 2, \ldots, P$, while $m(k)$ and 
$g(k)$ respectively identify the associated modality and group
memberships.

Consider the multivariate regression model
\[
\bm{y}_i = \bm{\alpha} + \bm{\Psi}^{\top}\bm{z}_i + 
\bm{B}^{\top}\bm{x}_i + \bm{\varepsilon}_i,
\]
where $\bm{\alpha} \in \mathbb{R}^Q$ is an intercept vector,
\[
\bm{\Psi} = (\bm{\psi}_1, \bm{\psi}_2, \ldots,\bm{\psi}_Q) \in
\mathbb{R}^{L\times Q}
\text{ with } \bm{\psi}_q\in\mathbb{R}^{L},
\]
is the matrix of regression coefficients corresponding to the
low-dimensional, non-regularized covariates,
\[
\bm{B} = (\bm{\beta}_1, \bm{\beta}_2, \ldots, \bm{\beta}_Q) \in 
\mathbb{R}^{P\times Q}
\text{ with } \bm{\beta}_q \in \mathbb{R}^{P},
\]
is the matrix of regression coefficients corresponding to the
high-dimensional predictors, and $\bm{\varepsilon}_i
\sim \mathcal{N}_Q(\bm{0},\bm{\Sigma})$ captures the residual 
dependence among outcomes.

The coefficient matrix $\bm{B}$ is assumed to be sparse, reflecting
that only a subset of the high-dimensional predictors contributes
meaningfully to the psychometric outcomes. In contrast, the
low-dimensional covariates in $\bm{z}_i$ are treated as adjustment 
variables and are not subject to shrinkage. The primary
objective of the proposed framework is to estimate the regression
coefficient matrix $\bm{B}$ and identify important predictors while
accounting for both the hierarchical organization of the
high-dimensional predictors and the multivariate dependence among the
outcomes.

\subsection{Multivariate Regression with Missing Outcomes}
\label{sec:missing}

Conditional on the predictor and covariate vectors $\bm{x}_i$ and 
$\bm{z}_i$ and model parameters $\bm{\theta} := (\bm{\alpha}, 
\bm{\Psi}, \bm{B}, \bm{\Sigma})$, the proposed multivariate 
regression model assumes
\begin{equation}
	\label{eq:mvn_model}
	\bm{y}_i \mid \bm{x}_i, \bm{z}_i, \bm{\theta} \sim
	\mathcal{N}_Q \left(\bm{\alpha} + \bm{\Psi}^{\top}\bm{z}_i
	+ \bm{B}^{\top}\bm{x}_i, \bm{\Sigma} \right).
\end{equation}
This formulation allows the covariance matrix $\bm{\Sigma}$ to 
capture residual correlations among psychometric domains and 
borrow information across outcomes when some outcomes are partially
observed.

Throughout the analysis, we assume that the missing data mechanism
satisfies the missing at random (MAR) condition; that is
\[
\Pr(\bm{r}_i \mid \bm{y}_i, \bm{x}_i, \bm{z}_i) =
\Pr(\bm{r}_i \mid \bm{y}_{i,\mathrm{obs}}, \bm{x}_i, \bm{z}_i),
\]
under which the probability of missingness depends on the observed
outcomes together with the observed covariates and predictors. Under 
this assumption, inference for $\bm{\theta}$ can be based on the 
observed-data likelihood obtained 
by integrating over the missing outcomes. Specifically, the 
contribution of subject $i$ to the observed-data likelihood is
\[
p(\bm{y}_{i,\mathrm{obs}} \mid \bm{x}_i, \bm{z}_i, \bm{\theta})
 = \int p(\bm{y}_i \mid \bm{x}_i, \bm{z}_i, 
\bm{\theta}) \, \mathrm{d} \bm{y}_{i,\mathrm{mis}},
\]
where $p(\bm{y}_i \mid \bm{x}_i, \bm{z}_i, \bm{\theta})$ is specified 
in Equation~\eqref{eq:mvn_model}.

Rather than directly maximizing the observed-data likelihood, which
requires computationally intensive numerical integration, we treat the
missing outcomes $\bm{y}_{i, \mathrm{mis}}$ as latent variables and
jointly infer them along with the model 
parameters~\cite{tanner1987calculation}. Specifically,
inference is based on the joint posterior distribution
$p(\bm{\theta}, \bm{Y}_{\mathrm{mis}} \mid \bm{Y}_{\mathrm{obs}}, 
\bm{X})$, where $\bm{Y}_{\mathrm{mis}} = \{\bm{y}_{1,\mathrm{mis}}, 
\bm{y}_{2,\mathrm{mis}}, \ldots, \bm{y}_{n,\mathrm{mis}}\}$ and 
$\bm{Y}_{\mathrm{obs}} = \{\bm{y}_{1,\mathrm{obs}}, 
\bm{y}_{2,\mathrm{obs}}, \ldots,\bm{y}_{n,\mathrm{obs}}\}$
denote the collections of missing and observed outcomes,
respectively. This latent-variable formulation propagates uncertainty
associated with the missing outcomes directly into posterior
inference for all model parameters. By jointly inferring the 
missing outcomes and model parameters within a unified Bayesian 
framework, this approach avoids potential incompatibility 
between separate imputation and analysis stages, providing coherent 
Bayesian inference for multivariate variable selection 
\citep{gelman2013bayesian}.

As a natural consequence of the joint Bayesian formulation, posterior
draws of the missing outcomes additionally generate completed
datasets, yielding a posterior-based multiple imputation
representation of the missing data. Under the multivariate
regression model in Equation~\eqref{eq:mvn_model}, posterior 
inference for the model parameters is performed directly from their 
joint posterior distribution and therefore does not require a 
separate multiple imputation stage. Instead, the completed datasets 
provide a coherent and convenient interface for downstream analyses 
using statistical models other than the proposed multivariate 
regression model that require fully observed outcome data. Additional 
details are presented in Section~\ref{sec:mi}.

\subsection{Hierarchical Horseshoe Prior Specification}
\label{sec:hhorseshoe}

To perform multilevel variable selection aligned with the 
hierarchical organization of the predictors, we propose SHIM 
(\underline{S}hrinkage for \underline{H}ierarchical predictors and 
\underline{I}ncomplete \underline{M}ultivariate outcomes), which 
employs a hierarchical horseshoe prior on the regression coefficient 
matrix $\bm{B}$, motivated by grouped and hierarchical extensions of 
the classical horseshoe prior~\cite{alt2025hierarchical, 
xu2016bayesian}. 
Let $\beta_{kq}$ denote the regression coefficient corresponding to
predictor $k$ and outcome $q$. We introduce a hierarchical
shrinkage structure that encourages predictor-specific sparsity to
be shared across outcomes while allowing outcome-specific
deviations. Specifically, we assume
\begin{equation}
	\label{eq:hs_beta}
	\beta_{kq} \mid \tau, \gamma_{m(k)}, \delta_{m(k), g(k)}, \lambda_k, 
	\xi_{kq} \sim \mathcal{N} \left(0, \tau^2 \gamma_{m(k)}^2 
	\delta_{m(k), g(k)}^2 \lambda_k^2	\xi_{kq}^2 \right),
\end{equation}
where $\tau > 0$ is a global shrinkage parameter controlling the
overall sparsity across predictors; $\gamma_{m(k)} > 0$ is a 
modality-level shrinkage parameter for the modality containing 
predictor $k$, allowing entire modalities to be adaptively shrunk 
toward zero when weakly associated with the outcomes; $\delta_{m(k), g(k)} 
> 0$ is a group-level shrinkage parameter for the group
containing predictor $k$, inducing regularization among related
predictors within the same modality; $\lambda_k > 0$ is a 
predictor-level local shrinkage parameter shared across all outcomes; 
and $\xi_{kq} > 0$ is an outcome-specific local shrinkage parameter 
that allows individual coefficients to deviate from the shared 
sparsity pattern. These shrinkage parameters jointly induce
regularization at the global, modality, group, predictor, and
outcome levels. Note that the indices $m(k)$ and $g(k)$ identify
the modality and the associated within-modality group of
predictor $k$, respectively.

We assign independent half-Cauchy priors to the shrinkage parameters:
\[
\tau \sim \mathcal{C}^+(0, 1), \quad
\gamma_m \sim \mathcal{C}^+(0, 1), \quad
\delta_g \sim \mathcal{C}^+(0, 1), \quad
\lambda_k \sim \mathcal{C}^+(0, 1), \quad
\xi_{kq} \sim \mathcal{C}^+(0, 1),
\]
for all $m = 1, 2, \ldots, M$, $g = 1, 2, \ldots, G_m$, $k = 1, 2, 
\ldots, P$, and $q = 1, 2, \ldots, Q$. This specification extends the 
classical horseshoe prior \cite{carvalho2010thehorseshoe}
to a hierarchical setting that enables adaptive shrinkage across
multiple levels of predictor organization.

We would like to note that the shrinkage parameters in 
Equation~\eqref{eq:hs_beta} enter the model through a multiplicative 
variance decomposition, and are therefore not individually 
identifiable. Specifically, different configurations of
$(\tau, \gamma_{m(k)}, \delta_{m(k), g(k)}, \lambda_k, \xi_{kq})$
may induce the same marginal variance for $\beta_{kq}$, resulting in 
a non-unique decomposition of the overall shrinkage effect. Such 
non-identifiability is a common characteristic of global-local 
shrinkage priors and has been widely discussed in the literature
\citep{bhadra2020horseshoe, bhadra2019lasso, 
carvalho2010thehorseshoe}. However, this non-identifiability does not 
affect inference on the regression coefficients $\bm{B}$, which are 
the primary quantities of interest. Accordingly, the shrinkage 
parameters should be viewed as auxiliary variables that jointly 
determine the amount of regularization through their product, rather 
than as separately interpretable model parameters.

\subsection{Posterior Computation}
\label{sec:posterior}

Statistical inference is based on the joint posterior distribution
\[
p(\bm{\theta}, \bm{Y}_{\mathrm{mis}}, \tau, \bm{\gamma},
\bm{\delta}, \bm{\lambda}, \bm{\xi} \mid \bm{Y}_{\mathrm{obs}}, 
\bm{X}, \bm{Z}),
\]
where $\bm{\gamma}$, $\bm{\delta}$, $\bm{\lambda}$, and 
$\bm{\xi}$ collect the shrinkage parameters introduced in 
Section~\ref{sec:hhorseshoe}. Direct evaluation of this
posterior distribution is analytically intractable due to the
hierarchical shrinkage structure and the presence of missing outcomes.
We therefore adopt a Bayesian data augmentation 
strategy~\cite{tanner1987calculation} and construct
a Markov chain Monte Carlo (MCMC) algorithm that iteratively samples
from the full conditional distributions of each block of unknown
quantities. Specifically, at each iteration, the algorithm first
updates the missing outcomes $\bm{Y}_{\mathrm{mis}}$ conditional on 
the current parameter values, followed by the regression parameters
$(\bm{\alpha}, \bm{\Psi}, \bm{B})$, the covariance matrix 
$\bm{\Sigma}$, and the shrinkage parameters $(\tau, \gamma_{m(k)}, 
\delta_{m(k), g(k)}, \lambda_k, \xi_{kq})$. These updates constitute a 
Gibbs sampler with closed-form or conjugate distributions. For 
notation simplicity, iteration superscripts are omitted throughout 
this subsection.

For each individual $i$, the conditional distribution of the missing
outcome components $\bm{y}_{i, \mathrm{mis}}$ given the observed 
components $\bm{y}_{i, \mathrm{obs}}$ and the current parameter 
values follows directly from the multivariate normal model in
Equation~\eqref{eq:mvn_model}. Define $\bm{\mu}_i = \bm{\alpha} + 
\bm{\Psi}^{\top}\bm{z}_i + \bm{B}^{\top}\bm{x}_i,$, and partition 
$\bm{y}_i$, $\bm{\mu}_i$, and $\bm{\Sigma}$ according to the observed 
and missing outcome components for subject $i$:
\[
\bm{y}_i = 
\begin{pmatrix}
	\bm{y}_{i, \mathrm{obs}} \\
	\bm{y}_{i, \mathrm{mis}}
\end{pmatrix},
\qquad
\bm{\mu}_i =
\begin{pmatrix}
	\bm{\mu}_{i, o} \\
	\bm{\mu}_{i, m}
\end{pmatrix},
\qquad
\bm{\Sigma} =
\begin{pmatrix}
	\bm{\Sigma}_{oo} & \bm{\Sigma}_{om} \\
	\bm{\Sigma}_{mo} & \bm{\Sigma}_{mm}
\end{pmatrix},
\]
where the subscripts $o$ and $m$ denote the observed and missing 
outcome components, respectively, for
individual $i$. The conditional distribution of $\bm{y}_{i, 
\mathrm{mis}}$ is therefore
\[
\bm{y}_{i,\mathrm{mis}} \mid \bm{y}_{i,\mathrm{obs}}, \bm{x}_i, 
\bm{z}_i, \bm{\theta} \sim \mathcal{N} \left(\bm{\mu}_{i,m} + 
\bm{\Sigma}_{mo}\bm{\Sigma}_{oo}^{-1} 
(\bm{y}_{i,\mathrm{obs}}-\bm{\mu}_{i,o}), \bm{\Sigma}_{mm}
- \bm{\Sigma}_{mo}\bm{\Sigma}_{oo}^{-1} \bm{\Sigma}_{om} \right).
\]
Because the partition into observed and missing outcome components is
subject-specific, the dimensions of $\bm{y}_{i,\mathrm{mis}}$
and the associated covariance blocks generally vary across 
individuals.

To update the regression coefficients, we stack the completed outcome
data and perform a joint update. Denote $\bm{Y} = (\bm{y}_1, 
\bm{y}_2, \ldots, \bm{y}_n)^\top \in \mathbb{R}^{n \times Q}$ and 
$\bm{X} = (\bm{x}_1, \bm{x}_2, \ldots, \bm{x}_n)^\top \in 
\mathbb{R}^{n \times P}$. The regression model in 
Equation~\eqref{eq:mvn_model} can be written in matrix form as
\[
\bm{Y} = \bm{1}_n \bm{\alpha}^\top + \bm{Z}\bm{\Psi} + \bm{X}\bm{B} + 
\bm{E},
\]
where $\bm{1}_n$ is the $n$-dimensional vector of ones, and the rows
of $\bm{E}$ are independently distributed as $\mathcal{N}_Q(\bm{0}, 
\bm{\Sigma})$. Next, let $\bm{b}_k = (\beta_{k1}, \beta_{k2}, \ldots, 
\beta_{kQ})^\top \in \mathbb{R}^Q$, for $k = 1, 2, \ldots,P$, denote 
the coefficient vector associated with predictor $k$ across the $Q$ 
outcomes. Equivalently, $\bm{b}_k^\top$ is the $k$th row of $\bm{B}$. 
Conditional on the shrinkage parameters, these predictor-specific
coefficient vectors are mutually independent and satisfy
\[
\bm{b}_k \mid \tau, \gamma_{m(k)}, \delta_{m(k), g(k)}, \lambda_k,
\xi_{k1}, \ldots, \xi_{kQ} \sim \mathcal{N}_Q(\bm{0}, \bm{\Omega}_k),
\]
with
\[
\bm{\Omega}_k =
\mathrm{diag}
\left(\tau^2\gamma_{m(k)}^2\delta_{m(k), g(k)}^2
\lambda_k^2 \xi_{k1}^2, \tau^2\gamma_{m(k)}^2\delta_{m(k), g(k)}^2
\lambda_k^2 \xi_{k2}^2, \ldots, \tau^2 \gamma_{m(k)}^2 \delta_{m(k), g(k)}^2
\lambda_k^2 \xi_{kQ}^2 \right).
\]

Let $\mathrm{vec}(\cdot)$ denote column-wise vectorization. Then
$\mathrm{vec}(\bm{B}^\top) = (\bm{b}_1^\top, \bm{b}_2^\top, \ldots, 
\bm{b}_P^\top)^\top$ is equivalent to stacking the rows of $\bm{B}$. 
Conditional on the shrinkage parameters, the prior for $\bm{B}$ 
admits the equivalent representation
\[
\mathrm{vec}(\bm{B}^\top) \mid \tau, \bm{\gamma}, \bm{\delta}, 
\bm{\lambda}, \bm{\xi}
\sim
\mathcal{N}_{PQ}(\bm{0},\bm{D}_{\bm{B}}),
\]
with $\bm{D}_{\bm{B}} = \mathrm{blockdiag}(\bm{\Omega}_1, 
\bm{\Omega}_2, \ldots, \bm{\Omega}_P)$.

Combining the matrix form of the Gaussian complete-data likelihood 
with the Gaussian prior for $\mathrm{vec}(\bm{B}^\top)$, the 
conditional posterior distribution of $\mathrm{vec}(\bm{B}^\top)$ is
\[
\mathrm{vec}(\bm{B}^\top) \mid \mathrm{others} \sim \mathcal{N}_{PQ}
(\bm{\mu}_{\bm{B}}, \bm{\Sigma}_{\bm{B}}),
\]
with posterior covariance and mean respectively given by
\[
\bm{\Sigma}_{\bm{B}} = \left\{\bm{D}_{\bm{B}}^{-1} + (\bm{X}^\top 
\bm{X}) \otimes \bm{\Sigma}^{-1}\right\}^{-1}, \qquad 
\bm{\mu}_{\bm{B}} = \bm{\Sigma}_{\bm{B}} (\bm{X}^\top \otimes 
\bm{\Sigma}^{-1}) \mathrm{vec}\left\{(\bm{Y} - \bm{1}_n 
\bm{\alpha}^\top - \bm{Z}\bm{\Psi})^{\top}\right\};
\]
see Appendix~\ref{app:conditionals_B} for detailed derivations.

To update the intercept vector $\bm{\alpha}$ and the regression
coefficient matrix $\bm{\Psi}$ corresponding to the low-dimensional
covariates that are not subject to shrinkage, we combine them into a
single coefficient matrix. Specifically, define
$\widetilde{\bm{Z}} = (\bm{1}_n, \bm{Z}) \in
\mathbb{R}^{n \times (L + 1)}$ and
\[
\widetilde{\bm{\Psi}}
=
\begin{pmatrix}
	\bm{\alpha}^{\top} \\
	\bm{\Psi}
\end{pmatrix}
\in \mathbb{R}^{(L + 1) \times Q},
\]
where the first row of $\widetilde{\bm{\Psi}}$ corresponds to the
intercept vector and the remaining rows correspond to the regression
coefficients associated with the low-dimensional covariates. The
multivariate regression model can therefore be written as
\[
\bm{Y} = \widetilde{\bm{Z}}\widetilde{\bm{\Psi}} + \bm{X}\bm{B} + 
\bm{E}.
\]

Let $\widetilde{\bm{\psi}}_q \in \mathbb{R}^{L + 1}$ denote the
$q$-th column of $\widetilde{\bm{\Psi}}$, corresponding to the
intercept and low-dimensional covariate coefficients for outcome
$q$. We assign independent Gaussian priors to the columns of
$\widetilde{\bm{\Psi}}$ as
\[
\widetilde{\bm{\psi}}_q \overset{\mathrm{ind}}{\sim} \mathcal{N}_{L + 
1}\left( \widetilde{\bm{\psi}}_{0q}, \widetilde{\bm{V}}_0
\right),
\]
for $q = 1, 2, \ldots, Q$, where
$\widetilde{\bm{\psi}}_{0q} \in \mathbb{R}^{L + 1}$ denotes the
prior mean vector for outcome $q$ and $\widetilde{\bm{V}}_0 \in
\mathbb{R}^{(L + 1) \times (L + 1)}$ is a symmetric positive
definite prior covariance matrix. For notation simplicity, define
$\widetilde{\bm{\Psi}}_0 = (\widetilde{\bm{\psi}}_{01}, 
\widetilde{\bm{\psi}}_{02}, \ldots, \widetilde{\bm{\psi}}_{0Q}
)\in \mathbb{R}^{(L + 1)\times Q}$. It follows that the prior for 
$\widetilde{\bm{\Psi}}$ can be expressed as
\[
\mathrm{vec}(\widetilde{\bm{\Psi}}^\top) \sim \mathcal{N}_{(L + 1)Q}
\left(\mathrm{vec}(\widetilde{\bm{\Psi}}_0^\top),
\widetilde{\bm{V}}_0 \otimes \bm{I}_Q \right).
\]

Combining the matrix form of the Gaussian complete-data likelihood
with this Gaussian prior, the conditional posterior distribution of
$\mathrm{vec}(\widetilde{\bm{\Psi}}^\top)$ is 
\[
\mathrm{vec}(\widetilde{\bm{\Psi}}^\top) \mid \mathrm{others}
\sim \mathcal{N}_{(L + 1)Q} \left(\bm{\mu}_{\widetilde{\bm{\Psi}}},
\bm{\Sigma}_{\widetilde{\bm{\Psi}}} \right),
\]
with posterior covariance and mean respectively given by
\begin{align*}
	\bm{\Sigma}_{\widetilde{\bm{\Psi}}} &= \left\{
	\widetilde{\bm{V}}_0^{-1}\otimes\bm{I}_Q + 
	(\widetilde{\bm{Z}}^\top\widetilde{\bm{Z}})
	\otimes\bm{\Sigma}^{-1}\right\}^{-1}, \\
	\bm{\mu}_{\widetilde{\bm{\Psi}}} &= 
	\bm{\Sigma}_{\widetilde{\bm{\Psi}}} 
	\left\{(\widetilde{\bm{V}}_0^{-1}\otimes\bm{I}_Q) \mathrm{vec}
	(\widetilde{\bm{\Psi}}_0^\top) + 
	(\widetilde{\bm{Z}}^\top\otimes\bm{\Sigma}^{-1})\mathrm{vec}
	\left((\bm{Y} - \bm{X}\bm{B})^\top\right)\right\};
\end{align*}
See Appendix~\ref{app:conditionals_alpha} for detailed derivations.

To update the residual covariance matrix, define the residual vector
for subject $i$, conditional on the current values of
$\bm{\alpha}$, $\bm{\psi}$, $\bm{B}$, and the completed outcome data, 
by $\bm{e}_i = \bm{y}_i - \bm{\alpha} - \bm{\Psi}^{\top}\bm{z}_i - 
\bm{B}^{\top}\bm{x}_i$. We assign the 
inverse-Wishart prior $\bm{\Sigma} \sim \mathcal{IW}(\nu_0, 
\bm{S}_0)$, where $\nu_0$ denotes the prior degrees of freedom and
$\bm{S}_0$ is a symmetric positive definite matrix. Since the 
inverse-Wishart prior is conjugate to the multivariate
normal likelihood~\cite{gelman2013bayesian}, the posterior 
distribution of
$\bm{\Sigma}$ is
\[\bm{\Sigma} \mid \mathrm{others} \sim \mathcal{IW}
\left(\nu_0 + n, \bm{S}_0 + \sum_{i = 1}^{n}\bm{e}_i\bm{e}_i^\top
\right).
\]

Finally, the shrinkage parameters introduced in
Section~\ref{sec:hhorseshoe} are updated using the auxiliary-variable
representation of the half-Cauchy prior proposed
by~\cite{makalic2016simple}. Under this representation, each
half-Cauchy prior is expressed as an inverse-gamma mixture.
Together with the Gaussian prior for the regression
coefficients in Equation~\eqref{eq:hs_beta}, this leads to
inverse-gamma conditional posterior distributions for the
corresponding squared shrinkage parameters. Within each Gibbs 
sampling iteration, the global, modality-level,
group-level, predictor-level, and outcome-specific shrinkage
parameters are updated sequentially according to their respective 
conditional posterior distributions. The complete Gibbs sampling
procedure is summarized in Algorithm~\ref{alg:gibbs}.

\begin{algorithm}[tbp]
	\caption{Gibbs Sampler for Posterior Computation}
	\label{alg:gibbs}
	\begin{algorithmic}[1]
		\STATE Initialize $\bm{Y}_{\mathrm{mis}}^{(0)}$, 
		$\bm{\alpha}^{(0)}$, $\bm{\Psi}^{(0)}$, $\bm{B}^{(0)}$, 
		$\bm{\Sigma}^{(0)}$, and 
		shrinkage parameters.
		\FOR{$t = 1, 2, \ldots, T$}
		\STATE Sample $\bm{Y}_{\mathrm{mis}}^{(t)}$ from its 
		conditional	multivariate normal distribution.
		\STATE Sample $\bm{B}^{(t)}$ from its conditional 
		multivariate normal distribution.
		\STATE Sample $\bm{\alpha}^{(t)}$ and $\bm{\Psi}^{(t)}$ 
		jointly	from their conditional multivariate normal 
		distribution.
		\STATE Sample $\bm{\Sigma}^{(t)}$ from its conditional 
		inverse-Wishart distribution.
		\STATE Update the shrinkage parameters
		$\tau^{(t)}$, $\bm{\gamma}^{(t)}$, $\bm{\delta}^{(t)}$,
		$\bm{\lambda}^{(t)}$, and $\bm{\xi}^{(t)}$
		from their respective conditional posterior distributions.
		\ENDFOR
		\STATE Retain posterior samples after burn-in for inference.
	\end{algorithmic}
\end{algorithm}

\subsection{Posterior-Based Multiple Imputation}
\label{sec:mi}

As discussed in Section~\ref{sec:missing}, inference for the model
parameters is conducted directly from their marginal posterior
distributions. Meanwhile, posterior samples of
$\bm{Y}_{\mathrm{mis}}$ also provide a natural posterior-based
multiple imputation representation, which is particularly useful when 
the completed outcome data are subsequently analyzed using 
statistical models that differ from Equation~\eqref{eq:mvn_model}.

Suppose that, after discarding the burn-in iterations, a total of
$T-T_{\mathrm{burnin}}$ posterior samples are retained. To
construct multiple imputations, a subset of $M$ posterior draws
($M \ll T-T_{\mathrm{burnin}}$) is selected, for example,
through systematic thinning. For $m=1,\ldots,M$, the $m$th selected
posterior draw of the missing outcomes,
$\bm{Y}_{\mathrm{mis}}^{(m)}$, together with the observed outcomes,
defines the completed dataset given by $\bm{Y}^{(m)} = 
(\bm{Y}_{\mathrm{obs}}, \bm{Y}_{\mathrm{mis}}^{(m)})$. The completed 
datasets $\{\bm{Y}^{(1)}, \bm{Y}^{(2)}, 
\ldots, \bm{Y}^{(M)}\}$ may subsequently be analyzed
using statistical models other than Equation~\eqref{eq:mvn_model} 
that require fully observed outcome data. These downstream analyses 
may additionally adjust for fully observed covariates, including 
those contained in $\bm{Z}$, depending on the scientific objective. 
In such settings, inference can be conducted using Rubin's 
rules~\cite{rubin1987multiple}. 

\begin{remark}
	It is worth emphasizing that the posterior-based multiple 
	imputation representation serves a different purpose from 
	posterior inference under the proposed Bayesian model. For the 
	multivariate regression model in Equation~\eqref{eq:mvn_model}, 
	inference for the model parameters is conducted directly from the 
	joint posterior distribution, without requiring a separate 
	multiple imputation stage. Consequently, parameter estimation and 
	uncertainty quantification are performed within a single Bayesian 
	framework, avoiding the need to specify separate imputation and 
	analysis models. The completed datasets are instead intended for 
	downstream analyses based on alternative statistical models that 
	require fully observed outcome data.
\end{remark}

\subsection{Theoretical Properties of the SHIM Prior}
\label{sec:theory}

We provide two theoretical characterizations of the SHIM prior. 
First, we characterize the dependence among coefficient magnitudes 
induced by the multiplicative shrinkage hierarchy. The following 
result shows that this dependence is determined by the number of 
shrinkage components shared by two coefficients, quantifying how the 
prior encourages related coefficients to exhibit similar magnitudes 
while retaining coefficient-specific variation.

\begin{proposition}
	\label{prop:dependence}
	Consider a multiplicative hierarchical shrinkage prior with
	independent positive random variables $\{H_\ell\}$. For 
	coefficient	$a$, let $\mathcal{P}_a$ be the finite set of 
	hierarchical components acting on that coefficient, and suppose
	\[
	\beta_a = s_aZ_a,
	\qquad
	s_a = \prod_{\ell \in \mathcal{P}_a}H_\ell,
	\]
	where $Z_a$ is a standard Gaussian variable that is independent 
	of the shrinkage hierarchy. Further, the Gaussian variables are 
	independent across distinct parameters. Assume $0 < 
	\operatorname{Var}\{\log(H_\ell)\} < \infty$ for every $\ell$. 
	Then, for two distinct coefficients $a \ne b$, we have
	\[
	\operatorname{Cov}\{\log|\beta_a|, \log|\beta_b|\} =
	\sum_{\ell \in \mathcal{P}_a \cap \mathcal{P}_b}
	\operatorname{Var}\{\log(H_\ell)\}
	\quad \text{and} \quad
	\operatorname{Var}\{\log|\beta_a|\} = \sum_{\ell \in \mathcal P_a}
	\operatorname{Var}\{\log(H_\ell)\} + \frac{\pi^2}{8}.
	\]
\end{proposition}

The identities in Proposition~\ref{prop:dependence} directly follow 
from
\[
\log|\beta_a| = \sum_{\ell \in \mathcal{P}_a}\log(H_\ell) + \log|Z_a|
\]
and the independence of the primitive shrinkage components and
Gaussian variables, together with the known result of 
$\operatorname{Var}(\log|Z_a|) =  \pi^2 / 8$. A detailed derivation 
is omitted.

For SHIM, the prior scale associated with $\beta_{kq}$ is $s_{kq}
= \tau\gamma_{m(k)}\delta_{m(k), g(k)}\lambda_k\xi_{kq}$. For a standard 
half-Cauchy variable $H$, we know $\operatorname{Var}(\log H) = \pi^2 
/ 4$. Because every coefficient-specific scale contains five 
independent half-Cauchy components, we get
\[
\operatorname{Var}\{\log|\beta_{kq}|\} = \frac{5 \pi^2}{4} + 
\frac{\pi^2}{8} = \frac{11 \pi^2}{8}.
\]
In what follows, we have
\[
\operatorname{Corr}
\{\log|\beta_{kq}|,\log|\beta_{k'q'}|\}
=
\begin{cases}
	8/11, & k = k',\ q \neq q',\\
	6/11, & k \neq k',\ m(k)=m(k'),\ g(k) = g(k'),\\
	4/11, & m(k) = m(k'),\ g(k) \neq g(k'),\\
	2/11, & m(k) \neq m(k').
\end{cases}
\]
These results imply that coefficients sharing more levels of the 
hierarchy have more strongly associated prior magnitudes. The shared 
global, modality, group, and predictor-specific scales induce this 
dependence, while the independent Gaussian variables and 
outcome-specific scales allow individual coefficients to retain 
distinct magnitudes and signs. This dependence provides a mechanism 
through which posterior learning about a shared scale can propagate 
across related coefficients.

We next characterize the marginal prior density induced when the 
conditional standard deviation of a Gaussian coefficient is the 
product of independent half-Cauchy variables. The result shows how 
the number of multiplicative components determines how rapidly the 
density diverges near zero and decays in the tails.

\begin{theorem}
	\label{thm:marginal_prior}
	Let $K \ge 1$ be an integer, and define
	\[
	H_1, H_2, \ldots , H_K \overset{\mathrm{ind}}{\sim}
	\mathcal C^{+}(0, 1), \qquad S_K=\prod_{\ell=1}^K H_\ell.
	\]
	Suppose	$\beta \mid S_K \sim \mathcal{N}(0, S_K^2)$	and let 
	$\pi_K$ denote the marginal density of $\beta$. Then, as
	$|\beta| \downarrow 0$, we have
	\[
	\pi_K(\beta) \sim C_{0, K} 
	\left\{\log\left(\frac{1}{|\beta|}\right)\right\}^K
	\quad \text{with} \quad
	C_{0, K} = \frac{(2/\pi)^K}{\sqrt{2\pi} K!}.
	\]
	Moreover, as $|\beta| \rightarrow \infty$, we have
	\[
	\pi_K(\beta) \sim C_{\infty, 
	K}\frac{\{\log|\beta|\}^{K-1}}{\beta^2}
	\quad \text{with} \quad
	C_{\infty,K} = \frac{(2/\pi)^K}{\sqrt{2\pi}\,(K-1)!}.
	\]
\end{theorem}

The proof of Theorem~\ref{thm:marginal_prior} is provided in
Appendix~\ref{app:marginal_prior}. For SHIM, the coefficient-specific 
scale is the product of $K = 5$ independent
half-Cauchy variables. The theorem therefore gives
\[
\pi_{\mathrm{SHIM}}(\beta)
\sim
\begin{cases}
	\displaystyle
	\frac{4}{15\pi^5\sqrt{2\pi}}
	\left\{\log\left(\frac{1}{|\beta|}\right)\right\}^5,
	& \text{as } |\beta|\downarrow 0,
	\\
	\displaystyle
	\frac{4}{3\pi^5\sqrt{2\pi}}
	\frac{\{\log|\beta|\}^4}{\beta^2},
	& \text{as } |\beta| \rightarrow \infty.
\end{cases}
\]

Thus, the multilevel construction produces an integrable
fifth-order logarithmic singularity at zero and a heavy polynomial
tail modified by a fourth-order logarithmic factor. The density is
continuous on $\mathbb{R} \setminus \{0\}$ but diverges as
$|\beta| \downarrow 0$, assigning substantial density to coefficients
near zero. In the tails, $\pi_{\mathrm{SHIM}}(\beta)$ has the same
polynomial decay exponent as a Cauchy density, while the factor
$\{\log|\beta|\}^4$ makes it asymptotically heavier than a standard
Cauchy density. These features encourage strong shrinkage of
small coefficients while maintaining substantial prior support for
large effects.

\section{Simulations}
\label{sec:sim}

\subsection{Data Generation and Simulation Setup}
\label{sec:gen}

This section presents simulation studies designed to evaluate the
variable selection performance of the proposed method in comparison
with several competing approaches. Motivated from the Vanderbilt 
Memory and Aging Project~\cite{moore2020lower}, we generated 
$n=300$ independent observations with $P \in \{200,400\}$ predictors 
and $Q = 2$ continuous outcomes. The predictors were organized 
according to a two-level hierarchical structure consisting of
$M = 2$ modalities, each containing $G_1 = G_2 = 4$ groups of 
predictors. The two modalities contained equal numbers of predictors, 
and all groups were of equal size. That is, for $P = 200$, each 
modality contained $100$ predictors and each group contained $25$ 
predictors; for $P = 400$, each modality contained $200$ predictors 
and each group contained $50$ predictors.

To consider predictors with a hierarchical correlation structure, the 
rows of $\bm{X}$ were generated independently according to $\bm{x}_i
\sim \mathcal{N}_P(\mathbf{0}, \bm{\Sigma}_X)$. Specifically, the 
covariance between predictors $k$ and $k^{'}$ was specified as
\[
\mathrm{Cov}(x_{ik},x_{ik^{'}}) = \sigma_{\mathrm{mod}}^2 \, 
\mathbb{I}\{m(k) = m(k^{'})\} + \sigma_{\mathrm{grp}}^2 \, 
\mathbb{I}\{m(k) = m(k^{'}), g(k) = g(k^{'})\} + \sigma_{\mathrm{res}}^2 \mathbb{I}\{k 
= k^{'}\},
\]
where $m(k)$ and $g(k)$ respectively denote the modality and group 
membership of predictor $k$, and $\mathbb{I}\{\cdot\}$ denotes the 
standard indicator function. In addition, $\sigma_{\mathrm{mod}}^2$ 
represents the shared modality-level variation, 
$\sigma_{\mathrm{grp}}^2$ represents the additional
variation shared by predictors within the same group, and
$\sigma_{\mathrm{res}}^2$ denotes the residual variance. Under this 
specification, predictors within the same group shared both modality- 
and group-level variation, predictors in different groups within the 
same modality shared only modality-level variation, and predictors 
from different modalities were independent. We further set 
$\sigma_{\mathrm{mod}} = 0.8$, $\sigma_{\mathrm{grp}} = 0.6$, and 
$\sigma_{\mathrm{res}} = 0.8$, inducing correlations of approximately 
$0.6$ between distinct predictors within the same group and $0.4$ 
between predictors in different groups within the same modality. 
Following generation, each predictor was standardized to have sample 
mean zero and unit variance to mitigate the influence of predictor 
scale on variable selection~\cite{piironen2017sparsity}.

The outcome was generated according to $\bm{Y} = \bm{1}_n 
\bm{\alpha}^{\top} + \bm{X}\bm{B} + \bm{E}$, where $\bm{\alpha}$ was 
set to $(0.4, 0.4)^\top$, and the rows of $\bm{E}$ were independently 
generated from $\mathcal{N}_Q(\bm{0}, \bm{\Sigma}_Y)$ with
\[\bm{\Sigma}_Y = 
\begin{pmatrix}
	1 & 0.3\\
	0.3 & 1
\end{pmatrix},
\]
which corresponds to a residual correlation of $0.3$ between the two
outcomes. For simplicity, the simulation studies did not include 
low-dimensional adjustment covariates $\bm{Z}$. Since these 
covariates are treated as non-penalized components, the Gaussian 
update for $\bm{\Psi}$ is conjugate and does not alter the 
hierarchical shrinkage mechanism or variable-selection properties of 
the proposed prior on $\bm{B}$. This design choice also ensured a 
fair comparison among all competing methods described in 
Section~\ref{sec:competing}, as not all methods allow adjustment 
covariates to be included as non-penalized terms.

The $P\times Q$ coefficient matrix $\bm{B}$ was assumed to be sparse.
In the benchmark setting, four predictors were assigned nonzero
coefficients with two signal strengths. Specifically, two predictors
had coefficient vectors $(1.2,1.2)^\top$, whereas the remaining two
had coefficient vectors $(0.6,0.6)^\top$. All other entries of
$\bm{B}$ were set to zero. To evaluate the impact of the 
modality-level hierarchy, we considered two configurations for the 
allocation of the nonzero coefficients:
\begin{enumerate}
	\item Under the ``1M2G'' configuration, the four active 
	predictors were distributed across two groups within the same 
	modality, with one group containing the two stronger signals and 
	the other containing the two weaker signals.
	\item Under the ``2M2G'' configuration, the active
	predictors were distributed across two groups belonging to two
	different modalities, with one group containing the two stronger 
	signals and the other containing the two weaker signals. 
\end{enumerate}

Finally, missing values were imposed on the generated outcomes. In the
benchmark setting, we considered a missing-at-random (MAR) mechanism,
under which the probability that outcome $q$ was missing depended on a
fully observed predictor through a logistic regression model. For each
simulated dataset, the intercept of the logistic regression model was
calibrated to achieve a target proportion $\pi_{\mathrm{miss}} \in 
\{0.2, 0.6\}$ of subjects with at least one missing outcome, 
corresponding to moderate and high levels of missingness, 
respectively.

\begin{table}[tbp]
	\centering
	\setlength{\tabcolsep}{10pt}
	\caption{Benchmark simulation settings (8 scenarios): Unless 
	otherwise specified, all settings used $n = 300$, a MAR 
	missing-data mechanism, and four active predictors}
	\label{tab:setup}
	\begin{tabular}{lcccccccc}
		\toprule
		Dimension & \multicolumn{4}{c}{$P=200$} &
		\multicolumn{4}{c}{$P=400$}\\
		\cmidrule(lr){2-5}\cmidrule(lr){6-9}
		Configuration &
		\multicolumn{2}{c}{1M2G} &
		\multicolumn{2}{c}{2M2G} &
		\multicolumn{2}{c}{1M2G} &
		\multicolumn{2}{c}{2M2G}\\
		\cmidrule(lr){2-3}\cmidrule(lr){4-5}
		\cmidrule(lr){6-7}\cmidrule(lr){8-9}
		$\pi_{\mathrm{miss}}$
		& 0.2 & 0.6
		& 0.2 & 0.6
		& 0.2 & 0.6
		& 0.2 & 0.6 \\
		\midrule
		Scenario
		& S1 & S2
		& S3 & S4
		& S5 & S6
		& S7 & S8 \\
		\bottomrule
	\end{tabular}
\end{table}

Table~\ref{tab:setup} summarizes the eight benchmark simulation
settings. Each setting was replicated $100$ times. Additional
implementation details for the Bayesian methods, including the number
of burn-in iterations, posterior samples retained for inference, and
MCMC convergence diagnostics, will be described subsequently. SHIM 
was implemented using the open-source R package \texttt{shim}, which 
also documents random number seeds and other information necessary 
to reproduce all simulations reported in this section.

\subsection{Methods for Comparison and Implementation Details}
\label{sec:competing}

We compared SHIM with several competing methods representing 
different strategies for modeling multivariate outcomes and 
performing variable selection. First, to evaluate the benefit of 
jointly modeling multiple outcomes, we implemented a univariate 
counterpart of SHIM, denoted by SHIM\_U, which applies the proposed 
hierarchical shrinkage model separately to each outcome while 
retaining the same multilevel shrinkage structure. We further 
included two modern Bayesian approaches: Group Inverse-Gamma Gamma 
(GIGG)~\cite{boss2024GIGG} and Multivariate Bayesian Sparse Group 
Selection with Spike-and-Slab Priors 
(MBSGSSS)~\cite{liquet2017bayesian}. GIGG employs continuous 
shrinkage through inverse-gamma--gamma (IG--Gamma) priors, with the 
horseshoe prior arising as a special case. MBSGSSS is a multivariate 
extension of the Bayesian sparse group selection method based on 
spike-and-slab priors~\cite{xu2015bayesian}. For all Bayesian 
methods, including SHIM, SHIM\_U, GIGG, and MBSGSSS, we ran a total 
of $2{,}000$ MCMC iterations and discarded the first $500$ 
iterations as burn-in. These settings were chosen based on 
convergence diagnostics, in which trace plots of the active 
regression coefficients exhibited stable mixing after burn-in and 
autocorrelations decayed sufficiently rapidly. Representative trace 
plots are provided in Appendix~\ref{app:trace}.

Lasso~\cite{tibshirani1996regression} and group 
lasso~\cite{yuan2006model} were included as conventional 
frequentist regularization methods. Since both methods were 
originally developed for univariate outcomes, they were applied 
separately to each outcome. Lasso imposes coefficient-wise 
regularization and therefore performs variable selection at the 
individual level, whereas group lasso (glasso) imposes group-wise 
regularization and performs variable selection at the group level. 
Because both methods produce shrinkage estimates for the selected 
predictors, we further refitted an ordinary least squares (OLS) regression 
model using the selected predictors to obtain post-selection point 
estimates with reduced bias~\cite{belloni2013least}.

Except for SHIM and SHIM\_U, the competing methods do not directly 
accommodate incomplete outcomes. These methods were therefore fitted 
using a common complete-case sample consisting of observations with 
all observed data. In contrast, SHIM and SHIM\_U were fitted using 
all available observations through the proposed formulation for 
incomplete outcomes. Furthermore, GIGG, lasso, and glasso do 
not jointly model multivariate outcomes and were consequently fitted 
separately to each outcome. Finally, although glasso, GIGG, and 
MBSGSSS incorporate group-level structure, they do not accommodate 
the additional modality-level hierarchy considered by SHIM. 
Accordingly, only group membership information, but not modality 
membership information, was provided to these methods.

Variable selection was determined according to the estimation 
framework of each method. For continuous shrinkage methods, including 
SHIM, SHIM\_U, and GIGG, a predictor was considered selected for a 
given outcome if the corresponding 95\% marginal posterior credible 
interval excluded zero~\cite{boss2024GIGG, pas2017uncertainty}. For 
MBSGSSS, a predictor was considered selected if the posterior median 
of its regression coefficient was nonzero, following the general 
recommendation for spike-and-slab priors~\cite{liquet2017bayesian}. 
For lasso and glasso, predictors with nonzero estimated 
regression coefficients were considered selected. Table~\ref{tab:methods} summarizes the 
competing methods and their main characteristics.

\begin{table}[tbp]
\centering
\caption{Summary of the competing methods and their basic features}
\label{tab:methods}
\setlength{\tabcolsep}{10pt}
\begin{tabular}{lccc}
\toprule
Method & Modeling & Shrinkage Hierarchy \\
\midrule
SHIM & Multivariate & Multi-group \\
SHIM\_U & Univariate & Multi-group \\
Lasso & Univariate & Individual \\
Glasso & Univariate & Group \\
GIGG & Univariate & Group \\
MBSGSSS & Multivariate & Group \\
\bottomrule
\end{tabular}
\end{table}

\subsection{Evaluation Criteria}
\label{sec:eval}

Performance was evaluated with respect to variable selection and 
coefficient estimation. For all penalized methods (except MLM), we 
summarized variable-selection performance separately for each 
outcome using the true positive rate (TPR; sensitivity), the false 
positive rate (FPR), the false discovery rate (FDR; $1 - 
\text{precision}$), and the F1 score. The TPR was defined as the 
proportion of truly active predictors that were correctly selected, 
whereas the FPR was defined as the proportion of truly inactive 
predictors that were incorrectly selected. The FDR was defined as 
the proportion of selected predictors that were false positives, and 
the F1 score was defined as the harmonic mean of sensitivity and 
precision.

We further evaluated coefficient estimation performance for the four 
active predictors. We calculated the absolute percentage bias and 
empirical standard deviation of the coefficient estimates across 
simulation replicates. The empirical standard deviation was used to 
assess the efficiency of the estimators. To evaluate the calibration 
of uncertainty quantification, we calculated the ratio of the 
empirical standard deviation to the average standard error across 
simulation replicates. Finally, coverage probability was defined as 
the proportion of 95\% credible or confidence intervals containing 
the corresponding true coefficient and was compared with the nominal 
95\% level.

\subsection{Simulation Results}
\label{sec:sim_res}

\begin{table}[tbp]
	\centering
	\scriptsize
	\caption{Average variable-selection performance over $100$ 
	simulation replicates. Results are reported separately for each 
	outcome under the eight simulation settings. Performance is 
	evaluated using the true positive rate (TPR), false positive rate 
	(FPR), false discovery rate (FDR), and F1 score.}
	\label{tab:sel}
	\setlength{\tabcolsep}{3.37pt}
	\begin{tabular}{lcccccccccccccccc}
	\toprule
	& \multicolumn{8}{c}{S1} & \multicolumn{8}{c}{S2} \\
	\cmidrule(lr){2-9}
	\cmidrule(lr){10-17}
	& \multicolumn{4}{c}{Outcome 1}
	& \multicolumn{4}{c}{Outcome 2} 
	& \multicolumn{4}{c}{Outcome 1}
	& \multicolumn{4}{c}{Outcome 2} \\
	\cmidrule(lr){2-5}
	\cmidrule(lr){6-9}
	\cmidrule(lr){10-13}
	\cmidrule(lr){14-17}
	& TPR & FPR & FDR & F1
	& TPR & FPR & FDR & F1 
	& TPR & FPR & FDR & F1
	& TPR & FPR & FDR & F1 \\
	\midrule
	SHIM 
	& 1.00 & 0.00 & 0.01 & 0.99 
	& 1.00 & 0.00 & 0.01 & 0.99
	& 0.98 & 0.00 & 0.01 & 0.98 
	& 0.99 & 0.00 & 0.02 & 0.99 \\ 
	SHIM\_U 
	& 0.99 & 0.00 & 0.01 & 0.99 
	& 1.00 & 0.00 & 0.02 & 0.99
	& 0.96 & 0.00 & 0.00 & 0.97 
	& 0.98 & 0.00 & 0.01 & 0.98 \\ 
	Lasso 
	& 1.00 & 0.03 & 0.52 & 0.63 
	& 1.00 & 0.03 & 0.51 & 0.64
	& 1.00 & 0.03 & 0.53 & 0.62 
	& 1.00 & 0.03 & 0.55 & 0.61 \\ 
	Glasso 
	& 1.00 & 0.44 & 0.95 & 0.09 
	& 1.00 & 0.45 & 0.95 & 0.09
	& 1.00 & 0.38 & 0.94 & 0.11 
	& 1.00 & 0.35 & 0.94 & 0.11 \\ 
	GIGG 
	& 0.99 & 0.00 & 0.01 & 0.99 
	& 1.00 & 0.00 & 0.01 & 0.99
	& 0.87 & 0.00 & 0.00 & 0.92 
	& 0.86 & 0.00 & 0.01 & 0.92 \\  
	MBSGSSS 
	& 1.00 & 0.00 & 0.03 & 0.98 
	& 1.00 & 0.00 & 0.03 & 0.99
	& 0.99 & 0.00 & 0.06 & 0.97 
	& 1.00 & 0.00 & 0.06 & 0.97 \\
	
	\midrule
	& \multicolumn{8}{c}{S3} & \multicolumn{8}{c}{S4} \\
	\cmidrule(lr){2-9}
	\cmidrule(lr){10-17}
	& \multicolumn{4}{c}{Outcome 1}
	& \multicolumn{4}{c}{Outcome 2} 
	& \multicolumn{4}{c}{Outcome 1}
	& \multicolumn{4}{c}{Outcome 2} \\
	\cmidrule(lr){2-5}
	\cmidrule(lr){6-9}
	\cmidrule(lr){10-13}
	\cmidrule(lr){14-17}
	& TPR & FPR & FDR & F1
	& TPR & FPR & FDR & F1 
	& TPR & FPR & FDR & F1
	& TPR & FPR & FDR & F1 \\
	\midrule
	SHIM 
	& 1.00 & 0.00 & 0.02 & 0.99 
	& 1.00 & 0.00 & 0.01 & 1.00
	& 0.98 & 0.00 & 0.01 & 0.99 
	& 0.98 & 0.00 & 0.01 & 0.98 \\ 
	SHIM\_U 
	& 1.00 & 0.00 & 0.01 & 0.99 
	& 1.00 & 0.00 & 0.01 & 1.00
	& 0.96 & 0.00 & 0.00 & 0.98 
	& 0.97 & 0.00 & 0.00 & 0.98 \\ 
	Lasso 
	& 1.00 & 0.02 & 0.48 & 0.67 
	& 1.00 & 0.02 & 0.50 & 0.66
	& 0.99 & 0.04 & 0.55 & 0.60 
	& 1.00 & 0.03 & 0.53 & 0.63 \\ 
	Glasso 
	& 1.00 & 0.47 & 0.95 & 0.09 
	& 1.00 & 0.47 & 0.95 & 0.09
	& 1.00 & 0.39 & 0.94 & 0.11 
	& 1.00 & 0.41 & 0.95 & 0.10 \\ 
	GIGG 
	& 0.99 & 0.00 & 0.01 & 0.99 
	& 1.00 & 0.00 & 0.01 & 0.99
	& 0.84 & 0.00 & 0.01 & 0.90 
	& 0.85 & 0.00 & 0.01 & 0.91 \\  
	MBSGSSS 
	& 1.00 & 0.00 & 0.01 & 0.99 
	& 1.00 & 0.00 & 0.02 & 0.99
	& 0.99 & 0.00 & 0.04 & 0.98 
	& 0.99 & 0.00 & 0.04 & 0.97 \\
	
	\midrule
	& \multicolumn{8}{c}{S5} & \multicolumn{8}{c}{S6} \\
	\cmidrule(lr){2-9}
	\cmidrule(lr){10-17}
	& \multicolumn{4}{c}{Outcome 1}
	& \multicolumn{4}{c}{Outcome 2} 
	& \multicolumn{4}{c}{Outcome 1}
	& \multicolumn{4}{c}{Outcome 2} \\
	\cmidrule(lr){2-5}
	\cmidrule(lr){6-9}
	\cmidrule(lr){10-13}
	\cmidrule(lr){14-17}
	& TPR & FPR & FDR & F1
	& TPR & FPR & FDR & F1 
	& TPR & FPR & FDR & F1
	& TPR & FPR & FDR & F1 \\
	\midrule
	SHIM 
	& 0.99 & 0.00 & 0.03 & 0.98 
	& 1.00 & 0.00 & 0.02 & 0.99
	& 0.98 & 0.00 & 0.03 & 0.97 
	& 0.98 & 0.00 & 0.03 & 0.97 \\ 
	SHIM\_U 
	& 0.99 & 0.00 & 0.02 & 0.98 
	& 1.00 & 0.00 & 0.00 & 1.00
	& 0.93 & 0.00 & 0.02 & 0.95 
	& 0.94 & 0.00 & 0.02 & 0.96 \\ 
	Lasso 
	& 1.00 & 0.02 & 0.60 & 0.56 
	& 1.00 & 0.02 & 0.60 & 0.56
	& 0.99 & 0.02 & 0.61 & 0.55 
	& 1.00 & 0.02 & 0.62 & 0.54 \\ 
	Glasso 
	& 1.00 & 0.33 & 0.97 & 0.06 
	& 1.00 & 0.34 & 0.97 & 0.06
	& 1.00 & 0.35 & 0.97 & 0.06 
	& 1.00 & 0.36 & 0.97 & 0.06 \\ 
	GIGG 
	& 0.98 & 0.00 & 0.00 & 0.99 
	& 1.00 & 0.00 & 0.00 & 1.00
	& 0.82 & 0.00 & 0.02 & 0.88 
	& 0.82 & 0.00 & 0.01 & 0.89 \\  
	MBSGSSS 
	& 1.00 & 0.00 & 0.02 & 0.99 
	& 1.00 & 0.00 & 0.02 & 0.99
	& 0.99 & 0.00 & 0.05 & 0.97 
	& 0.99 & 0.00 & 0.06 & 0.96 \\
	
	\midrule
	& \multicolumn{8}{c}{S7} & \multicolumn{8}{c}{S8} \\
	\cmidrule(lr){2-9}
	\cmidrule(lr){10-17}
	& \multicolumn{4}{c}{Outcome 1}
	& \multicolumn{4}{c}{Outcome 2} 
	& \multicolumn{4}{c}{Outcome 1}
	& \multicolumn{4}{c}{Outcome 2} \\
	\cmidrule(lr){2-5}
	\cmidrule(lr){6-9}
	\cmidrule(lr){10-13}
	\cmidrule(lr){14-17}
	& TPR & FPR & FDR & F1
	& TPR & FPR & FDR & F1 
	& TPR & FPR & FDR & F1
	& TPR & FPR & FDR & F1 \\
	\midrule
	SHIM 
	& 1.00 & 0.00 & 0.03 & 0.99 
	& 1.00 & 0.00 & 0.02 & 0.99
	& 0.97 & 0.00 & 0.04 & 0.96 
	& 0.97 & 0.00 & 0.03 & 0.97 \\ 
	SHIM\_U 
	& 1.00 & 0.00 & 0.01 & 0.99 
	& 0.99 & 0.00 & 0.01 & 0.99
	& 0.94 & 0.00 & 0.01 & 0.96 
	& 0.95 & 0.00 & 0.02 & 0.96 \\ 
	Lasso 
	& 1.00 & 0.02 & 0.56 & 0.59 
	& 1.00 & 0.02 & 0.57 & 0.58
	& 1.00 & 0.02 & 0.64 & 0.51 
	& 1.00 & 0.02 & 0.65 & 0.51 \\ 
	Glasso 
	& 1.00 & 0.41 & 0.97 & 0.05 
	& 1.00 & 0.40 & 0.97 & 0.05
	& 1.00 & 0.41 & 0.97 & 0.05 
	& 1.00 & 0.42 & 0.97 & 0.05 \\ 
	GIGG 
	& 0.99 & 0.00 & 0.00 & 0.99 
	& 0.99 & 0.00 & 0.00 & 0.99
	& 0.81 & 0.00 & 0.01 & 0.88 
	& 0.81 & 0.00 & 0.01 & 0.88 \\  
	MBSGSSS 
	& 1.00 & 0.00 & 0.02 & 0.99 
	& 1.00 & 0.00 & 0.02 & 0.99
	& 0.99 & 0.00 & 0.04 & 0.97 
	& 0.99 & 0.00 & 0.04 & 0.97 \\
	\bottomrule
	\end{tabular}
\end{table}

Table~\ref{tab:sel} summarizes the variable-selection performance 
across the eight simulation settings. Overall, SHIM consistently 
achieved the best or nearly best performance across all scenarios, 
maintaining high TPR together with near-zero FPR and FDR, resulting 
in F1 scores close to one. Compared with its univariate counterpart 
SHIM\_U, SHIM showed slightly improved TPR and F1 scores in the more 
challenging settings with higher missingness (S2, S4, S6, and S8), 
suggesting that jointly modeling correlated outcomes provides 
additional gains when information is limited. Among the competing 
methods, GIGG effectively controlled false positives but experienced 
a noticeable reduction in TPR under the more challenging scenarios, 
whereas MBSGSSS maintained high TPR and consistently achieved F1 
scores comparable to those of SHIM. In contrast, lasso selected 
substantially more inactive predictors, resulting in substantially
inflated FDR despite perfect or nearly perfect TPR, whereas glasso 
presented extremely high FPR and FDR and consequently poor F1 
scores. Overall, these results demonstrate that SHIM provides an 
effective balance between sensitivity and false positive control, 
yielding robust variable-selection performance across a wide range of 
simulation settings. 

Figure~\ref{fig:box} displays the distribution of the F1 scores 
across the $100$ simulation replicates. Consistent with the average 
results in Table~\ref{tab:sel}, SHIM achieved the highest median F1 
scores with consistently small variability across nearly all 
simulation settings, implying stable variable-selection 
performance. MBSGSSS also demonstrated competitive performance, with 
F1 distributions that were generally comparable to those of SHIM. In 
contrast, GIGG showed a progressive decline in F1 scores as the 
simulation settings became more difficult, whereas lasso and glasso 
consistently yielded substantially lower F1 scores due to a large 
amount of false 
discoveries.

\begin{figure}[tbp]
	\centering
    \caption{Distribution of F1 scores across $100$ simulation 
    replicates under the eight simulation settings}
    \includegraphics[width=\textwidth]{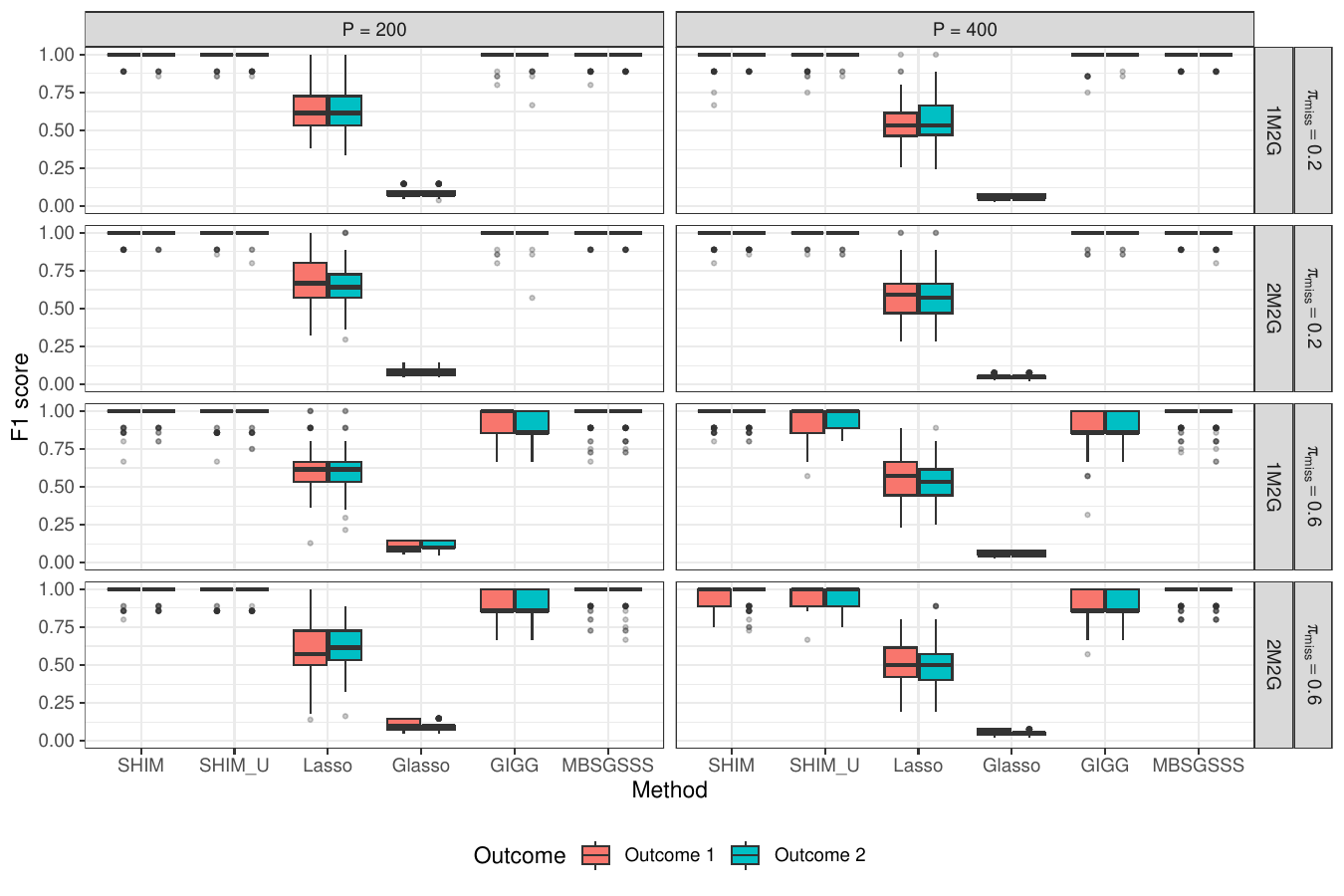}
    \label{fig:box}
\end{figure}

In addition, Figure~\ref{fig:est} presents the absolute percentage 
bias and coverage probabilities for the active regression 
coefficients. Across all methods, estimation was generally more 
accurate for the stronger signals ($\beta_1$ and $\beta_2$) than for 
the weaker signals ($\beta_3$ and $\beta_4$), particularly under high 
missingness (S2, S4, S6, and S8). SHIM, SHIM\_U, and MBSGSSS achieved 
comparably low bias for the stronger coefficients. For the weaker 
signals, bias increased for all three methods, although SHIM 
consistently yielded lower bias than its univariate counterpart 
SHIM\_U, particularly under the more challenging settings with larger 
predictor dimensions and higher missingness. MBSGSSS also 
demonstrated competitive performance and generally produced bias 
comparable to, or slightly smaller than, that of SHIM. In contrast, 
GIGG exhibited substantially larger bias for the weaker signals under 
high missingness, whereas lasso and glasso showed considerably 
larger bias across most simulation settings even though 
post-selection ordinary least squares estimation was implemented.

The coverage results followed similar patterns. SHIM, SHIM\_U, and 
MBSGSSS generally achieved coverage probabilities close to the 
nominal 95\% level across both outcomes, although some under-coverage 
occurred for the weaker coefficients under high missingness. Compared 
with SHIM\_U, SHIM more consistently maintained nominal coverage for 
the weaker signals in the more challenging settings. GIGG achieved 
approximately nominal coverage for the stronger coefficients but 
exhibited noticeable under-coverage for some weaker coefficients, 
particularly under high missingness. Lasso showed systematic 
under-coverage across most simulation settings. Although group lasso 
frequently attained near-nominal coverage among the available 
estimates, these results were accompanied by substantially greater 
empirical variability and less reliable standard error calibration; 
additional results on the empirical standard deviation and the ratio 
of empirical to estimated standard errors are provided in 
Appendix~\ref{app:add_sd}. Overall, SHIM provided accurate 
coefficient estimation together with well-calibrated interval 
estimation across a wide range of simulation settings.

\begin{figure}[tbp]
	\centering
	\caption{Absolute percentage bias (top) and coverage probability 
	(bottom) for the four active coefficients across the eight 
	simulation settings. Coefficients $\beta_1$ and $\beta_2$ 
	correspond to the two stronger signals (true coefficient $ = 
	1.2$), whereas $\beta_3$ and $\beta_4$ correspond to the two 
	weaker signals (true coefficient $ = 0.6$). 
	For lasso and group lasso, the reported summaries are based on 
	post-selection ordinary least squares estimates}
	\includegraphics[width=\textwidth]{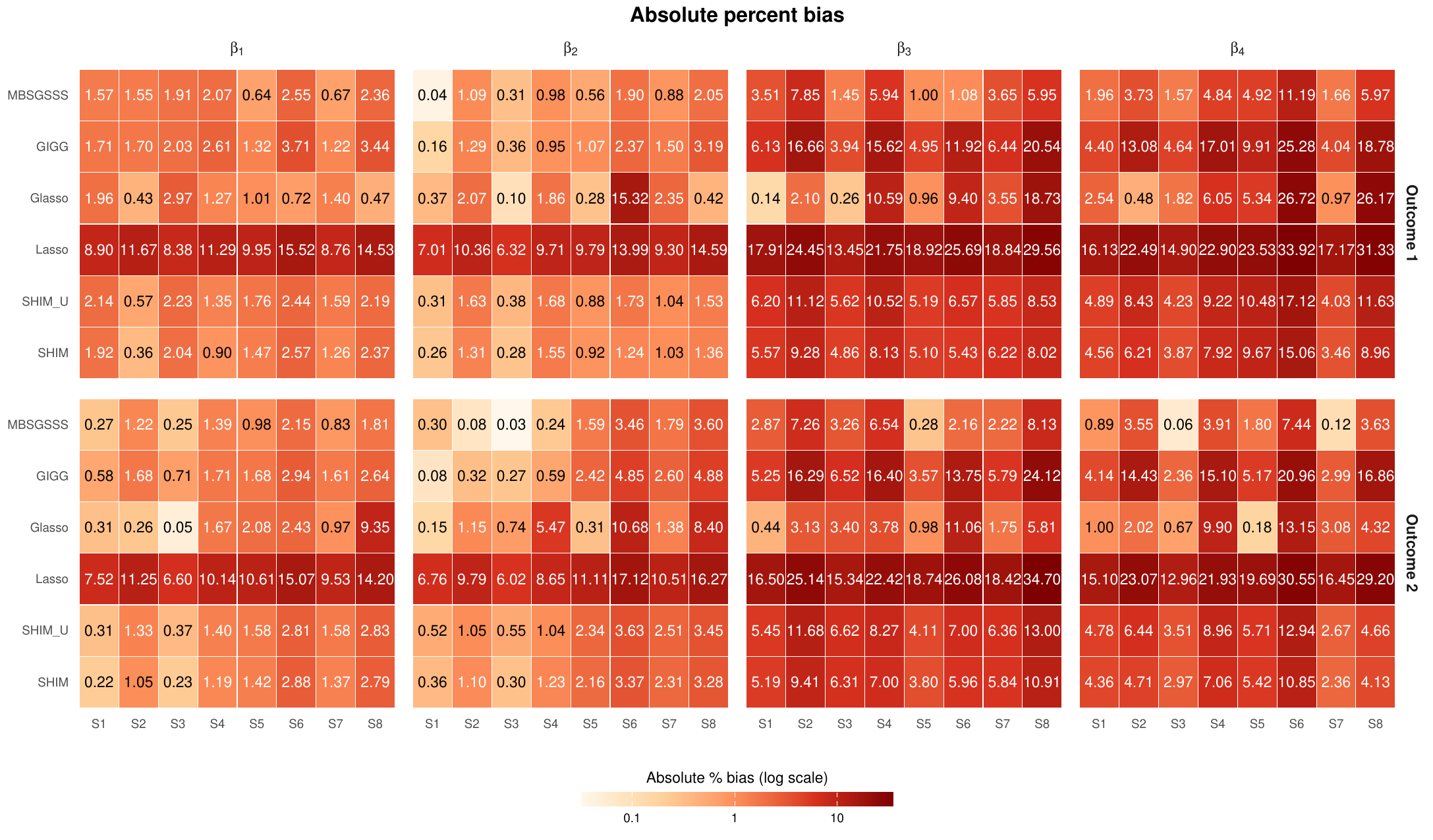}
	\includegraphics[width=\textwidth]{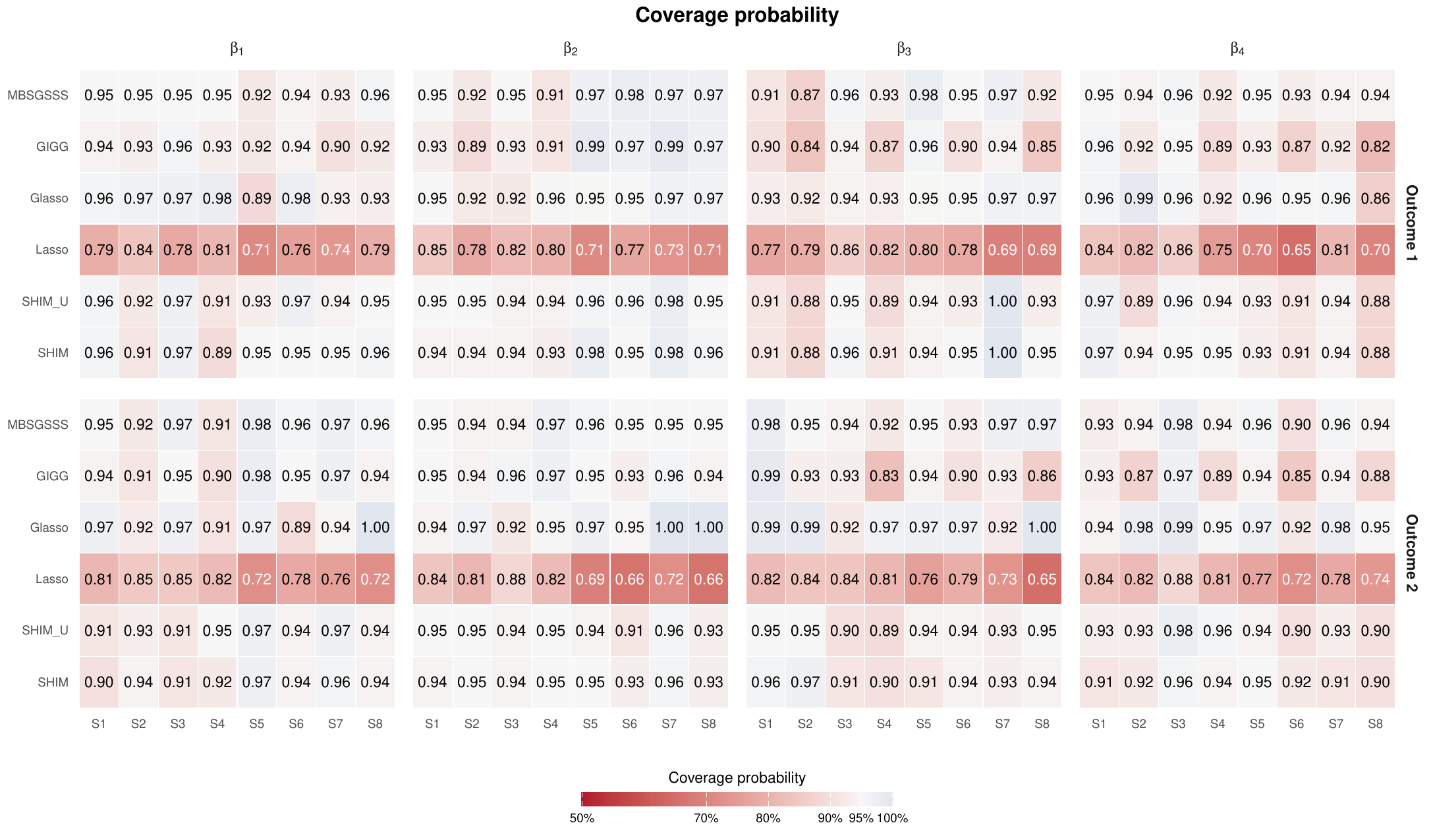}
	\label{fig:est}
\end{figure}

\subsection{Sensitivity Analysis}
\label{sec:sens}

To evaluate the robustness of the proposed framework under alternative
data-generating mechanisms, we considered four additional simulation
settings obtained by modifying S6 in Section~\ref{sec:gen}, detailed 
as follows:

\begin{enumerate}
	\item \textit{High outcome correlation (CORR)}: We increased the
	residual correlation between the two outcomes from $0.3$ to $0.7$.
	This setting evaluates the impact of stronger dependence among
	psychometric outcomes on variable selection and estimation.
	\item \textit{Missing not at random (MNAR)}: We considered an MNAR
	mechanism under which the probability that an outcome was missing
	depended on its potentially unobserved value together with the 
	value of the other outcome. The residual correlation between the 
	two	outcomes was set to $0.7$. This setting assesses the 
	robustness of the competing methods when the MAR assumption is 
	violated.
	\item \textit{Outcome-specific sparsity (OUT)}: We allowed the two
	outcomes to have different sets of active predictors. 
	Specifically, the two stronger signals had coefficient vectors
	$(1.2, 0)^\top$ and $(0, 1.2)^\top$, whereas the two weaker 
	signals had coefficient vectors $(0.6, 0)^\top$ and 
	$(0, 0.6)^\top$. This setting evaluates the ability of the method 
	to distinguish outcome-specific associations from predictors with 
	shared effects across outcomes. 
	\item \textit{Ultra-high dimensionality (ULTRA)}: We increased the
	number of predictors to $P = 600$, with two modalities containing 
	$300$ predictors each and each modality further divided into four 
	groups of $75$ predictors. This setting evaluates the scalability 
	of the competing methods under increasingly high-dimensional 
	predictors.
\end{enumerate}

Except for the modifications described above, all remaining
data-generating features were identical to those of S6, including 
$\pi_{\mathrm{miss}} = 0.6$. The same variable-selection and 
coefficient-estimation metrics used in the primary simulation study 
were calculated for each sensitivity setting.

The variable-selection results under CORR, MNAR, and ULTRA were 
generally consistent with those from the primary simulations. SHIM, 
SHIM\_U, and MBSGSSS continued to achieve comparably high F1 scores, 
whereas GIGG exhibited lower TPR values, and lasso and glasso 
continued to yield inflated FDR values; see Table~\ref{tab:sel_sens} 
in Appendix~\ref{app:res_sens} for details. Under OUT, however, 
MBSGSSS no longer performed comparably to SHIM. As shown in 
Table~\ref{tab:sel_out}, its average F1 scores decreased 
substantially, primarily because of increased FDR. This deterioration 
likely reflects the predictor-driven selection structure of MBSGSSS, 
under which sparsity is largely shared across outcomes and therefore 
may not adequately distinguish predictors associated with only a 
subset of responses. In contrast, SHIM incorporates outcome-specific 
local shrinkage and maintained strong selection performance under 
this particular scenario.

\begin{table}[tbp]
\centering
\caption{Average variable-selection performance under the
outcome-specific sparsity (OUT) setting over $100$ simulation replicates. Performance is evaluated using the true positive rate (TPR), false positive rate (FPR), false discovery rate (FDR), and F1 score}
\label{tab:sel_out}
\setlength{\tabcolsep}{9.5pt}
\begin{tabular}{lcccccccc}
\toprule
& \multicolumn{4}{c}{Outcome 1}
& \multicolumn{4}{c}{Outcome 2} \\
\cmidrule(lr){2-5}
\cmidrule(lr){6-9}
Method
& TPR & FPR & FDR & F1
& TPR & FPR & FDR & F1 \\
\midrule
SHIM
& 0.96 & 0.00 & 0.06 & 0.94
& 0.96 & 0.00 & 0.04 & 0.95 \\
SHIM\_U
& 0.95 & 0.00 & 0.01 & 0.96
& 0.94 & 0.00 & 0.01 & 0.95 \\
Lasso
& 1.00 & 0.01 & 0.48 & 0.65
& 0.99 & 0.01 & 0.53 & 0.60 \\
Glasso
& 0.97 & 0.33 & 0.98 & 0.03
& 0.98 & 0.37 & 0.99 & 0.03 \\
GIGG
& 0.84 & 0.00 & 0.00 & 0.89
& 0.80 & 0.00 & 0.01 & 0.86 \\
MBSGSSS
& 0.95 & 0.01 & 0.51 & 0.64
& 0.94 & 0.01 & 0.52 & 0.63 \\
\bottomrule
\end{tabular}
\end{table}

Coefficient-estimation patterns were generally similar to those in 
the primary simulations, although the OUT setting more clearly 
differentiated SHIM, SHIM\_U, and MBSGSSS. As shown in 
Table~\ref{tab:est_out}, MBSGSSS exhibited greater bias for weak 
signals, higher empirical variability, and more pronounced 
under-coverage. Its ratios of empirical to average estimated standard 
errors were also noticeably above one, suggesting that the estimated 
standard errors understated the sampling variability. This pattern 
may reflect difficulty in accommodating outcome-specific associations 
when selection of MBSGSSS is primarily structured at the predictor 
level. In contrast, SHIM and SHIM\_U generally provided accurate 
estimates with low empirical variability, although both showed 
increased bias for weak signals. SHIM also demonstrated a modest 
advantage over SHIM\_U in coverage and standard error calibration. 
Among the remaining methods, GIGG continued to exhibit greater bias 
and under-coverage for weak signals, lasso showed substantially 
greater bias and under-coverage across active coefficients, and 
glasso produced obviously greater empirical variability.

Across the other sensitivity settings, SHIM generally remained 
competitive in terms of bias and maintained relatively low empirical 
variability, particularly when the outcomes were highly correlated. 
Its performance deteriorated under MNAR, however, with increased bias 
and under-coverage. Under ULTRA, SHIM continued to perform reasonably 
well, although MBSGSSS appeared comparatively stronger. Full results 
and further discussion are provided in Table~\ref{tab:sel_sens} and \ref{tab:est_sens} of 
Appendix~\ref{app:res_sens}.

\begin{table}[tbp]
	\centering
	\caption{Estimation performance for the active predictors under OUT. Coefficients $\beta_1$ and $\beta_2$ 
	correspond to the two stronger signals (true coefficient $ = 
	1.2$), whereas $\beta_3$ and $\beta_4$ correspond to the two 
	weaker signals (true coefficient $ = 0.6$). Coefficients $\beta_1$ and $\beta_3$ are only associated with Outcome 1, whereas coefficients $\beta_2$ and $\beta_4$ are only associated with Outcome 2}
	\label{tab:est_out}
	\setlength{\tabcolsep}{3.3pt}
	\begin{tabular}{lcccccccccccccccc}
	\toprule
	& \multicolumn{16}{c}{Outcome 1} \\
	\cmidrule{2-17}
	& \multicolumn{4}{c}{PB (\%)}
	& \multicolumn{4}{c}{ESD}
	& \multicolumn{4}{c}{ESD / SE}
	& \multicolumn{4}{c}{CP (\%)} \\
	\cmidrule(lr){2-5}
	\cmidrule(lr){6-9}
	\cmidrule(lr){10-13}
	\cmidrule(lr){14-17}
	Method
	& $\beta_1$ & $\beta_2$ & $\beta_3$ & $\beta_4$
	& $\beta_1$ & $\beta_2$ & $\beta_3$ & $\beta_4$
	& $\beta_1$ & $\beta_2$ & $\beta_3$ & $\beta_4$
	& $\beta_1$ & $\beta_2$ & $\beta_3$ & $\beta_4$ \\
	\midrule

	SHIM
	& 2.6 &     & 7.4 &
	& 0.12 &      & 0.14 &
	& 0.94 &      & 1.03 &
	& 96 &    & 93 & \\

	SHIM\_U
	& 2.5 &     & 8.1 &
	& 0.11 &      & 0.14 &
	& 0.90 &      & 1.09 &
	& 94 &    & 90 & \\

	Lasso
	& 13.3 &      & 24.8 &
	& 0.17 &      & 0.14 &
	& 1.31 &      & 1.17 &
	& 73 &    & 69 & \\

	Glasso
	& 3.5 &      & 13.3 &
	& 0.42 &      & 0.46 &
	& 1.03 &      & 1.13 &
	& 94 &    & 87 & \\

	GIGG
	& 3.3 &      & 15.8 &
	& 0.16 &      & 0.19 &
	& 1.04 &      & 1.13 &
	& 93 &    & 90 & \\

	MBSGSSS
	& 3.8 &      & 13.6 &
	& 0.15 &      & 0.22 &
	& 1.01 &      & 1.37 &
	& 93 &    & 88 & \\
	\midrule
	
	& \multicolumn{16}{c}{Outcome 2} \\
	\cmidrule{2-17}
	& \multicolumn{4}{c}{PB (\%)}
	& \multicolumn{4}{c}{ESD}
	& \multicolumn{4}{c}{ESD / SE}
	& \multicolumn{4}{c}{CP (\%)} \\
	\cmidrule(lr){2-5}
	\cmidrule(lr){6-9}
	\cmidrule(lr){10-13}
	\cmidrule(lr){14-17}
	Method
	& $\beta_1$ & $\beta_2$ & $\beta_3$ & $\beta_4$
	& $\beta_1$ & $\beta_2$ & $\beta_3$ & $\beta_4$
	& $\beta_1$ & $\beta_2$ & $\beta_3$ & $\beta_4$
	& $\beta_1$ & $\beta_2$ & $\beta_3$ & $\beta_4$ \\
	\midrule

	SHIM
	&     & 3.8 &     & 11.6
	&     & 0.12 &     & 0.14
	&     & 1.00 &     & 1.04
	&     & 94 &     & 91 \\

	SHIM\_U
	&     & 3.6 &     & 12.6
	&     & 0.12 &     & 0.14
	&     & 1.02 &     & 1.09
	&     & 93 &     & 90 \\

	Lasso
	&     & 14.9 &     & 28.5
	&     & 0.15 &     & 0.16
	&     & 1.24 &     & 1.37
	&     & 68 &     & 64 \\

	Glasso
	&     & 5.7 &     & 7.0
	&     & 0.82 &     & 0.50
	&     & 2.05 &     & 1.24
	&     & 98 &     & 88 \\

	GIGG
	&     & 4.6 &     & 21.9
	&     & 0.13 &     & 0.22
	&     & 0.96 &     & 1.32
	&     & 94 &     & 85 \\

	MBSGSSS
	&     & 4.7 &     & 19.6
	&     & 0.12 &     & 0.23
	&     & 0.92 &     & 1.44
	&     & 94 &     & 85 \\
	\bottomrule
	\end{tabular}
\end{table}
\subsection{Posterior-Based Multiple Imputation for Downstream 
Analysis}
\label{sec:down}

To illustrate the use of posterior-based multiple imputation for
downstream analysis, we conducted an additional simulation in which a
partially observed variable $\bm{W}$ was subsequently used as a
predictor of a fully observed outcome $\bm{Y}$. Data were generated
according to
\[
\bm{W} = \alpha_W\bm{1}_n + \bm{X}\bm{\beta} + \bm{Z}\bm{\psi}_W
+ \bm{\varepsilon}_W \qquad \text{and} \qquad
\bm{Y} = \alpha_Y\bm{1}_n + \kappa\bm{W} + \bm{Z}\bm{\psi}_Y
+ \bm{\varepsilon}_Y
\]
with $\bm{\varepsilon}_W, \bm{\varepsilon}_Y \sim
\mathcal{N}_n(\bm{0}, \bm{I}_n)$, $\alpha_W = \alpha_Y = 0.4$, and 
$\kappa = 0.8$ representing the association of interest. The
high-dimensional predictors $\bm{X}$ and their coefficients were
generated following S6 (from Section~\ref{sec:gen}), while $\bm{Z}$ 
contained one standardized continuous covariate and one binary 
covariate with associated coefficients set to $\bm{\psi}_W = (0.5, 
0.3)^\top$ and $\bm{\psi}_Y = (0.6, 0.4)^\top$. Missing
values were imposed on $\bm{W}$ under the MAR mechanism with a target
missing proportion of $0.6$. This setting is motivated by applications
in which a partially observed biomarker, such as a cerebrospinal fluid
biomarker, is subsequently related to a fully observed psychometric
outcome.

SHIM was fitted jointly to $\bm{W}$ and $\bm{Y}$ conditional on
$\bm{X}$ and $\bm{Z}$, and $M=20$ completed datasets were obtained
from posterior draws of the missing values of $\bm{W}$. Within each
completed dataset, the downstream model regressed $\bm{Y}$ on
$\bm{W}$ and $\bm{Z}$ using ordinary least squares, and inference for
$\kappa$ was combined across imputations using Rubin's rules. The
procedure was repeated over $100$ simulation replicates to evaluate
estimation performance of $\kappa$.

Table~\ref{tab:down_res} summarizes the downstream estimation 
performance of posterior-based multiple imputation. For the 
association of interest ($\kappa$), the mean estimate showed low 
bias, and the average estimated standard error closely matched the 
empirical variability, although the confidence intervals exhibited 
modest under-coverage. The effects of the continuous and binary 
covariates ($\psi_{Y_1}$ and $\psi_{Y_2}$, respectively) were also 
estimated with reasonably low bias, with estimated standard errors 
comparable to their empirical counterparts and coverage probabilities 
close to the nominal $95\%$ level. Overall, these results suggest 
that posterior-based multiple imputation can support reasonably 
accurate downstream estimation and uncertainty quantification for 
both the association of interest and the covariate effects, even in 
this challenging setting with $60\%$ missing data.

\begin{table}[tbp]
\centering
\caption{Downstream estimation performance of posterior-based 
multiple imputation across $100$ simulation replicates, where 
$\kappa$ denotes the association of interest, whereas $\psi_{Y_1}$ 
and $\psi_{Y_2}$ denote the associations of the continuous and binary 
covariates with the outcome, respectively}
\setlength{\tabcolsep}{12.25pt}
\begin{tabular}{cccccc}
\toprule
Parameter & True Value & Estimate & ESD & ESD / SE & CP (\%) \\
\midrule
$\kappa$ & 0.8 & 0.78 & 0.03 & 0.95 & 90 \\
$\psi_{Y_1}$ & 0.6 & 0.64 & 0.10 & 0.98 & 94 \\
$\psi_{Y_2}$ & 0.4 & 0.41 & 0.16 & 0.94 & 95 \\
\bottomrule
\end{tabular}
\label{tab:down_res}
\end{table}

\section{Application}
\label{sec:app}

\subsection{Data Description}
\label{sec:data}

We applied SHIM to baseline data from the legacy cohort of the 
Vanderbilt Memory and Aging Project (VMAP), a longitudinal 
observational study designed to investigate the relationships among 
vascular health, brain aging, and cognitive impairment in older 
adults~\citep{moore2020lower}. VMAP collected comprehensive 
demographic, clinical, neuropsychological, fluid-biomarker, and 
multimodal neuroimaging data. The analytic sample included $288$ 
participants after excluding those with incomplete demographic 
information. Demographic and genetic characteristics of the analytic 
sample are summarized in Table~\ref{tab:vmap}.

\begin{table}[tbp]
	\centering
	\caption{Demographic and genetic characteristics of the VMAP
		analytic sample. NHW: non-Hispanic white; APOE-$\varepsilon 
		4$: apolipoprotein E epsilon 4}
	\label{tab:vmap}
	\setlength{\tabcolsep}{6.3pt}
	\begin{tabular}{lc}
		\toprule
		Characteristic & Overall ($n = 288$)\\
		\midrule
		Age, years & $72.8 \pm 7.4$\\
		Sex (female), \% & $43.4$\\
		Education, years & $15.9 \pm 2.7$\\
		Race/Ethnicity (NHW), \% & $87.2$ \\
		{\em APOE}-$\varepsilon 4$ carrier (positive), \% & $35.4$ \\
		\bottomrule
	\end{tabular}
\end{table}

We considered regional gray matter (GM) volume and cerebral blood 
flow (CBF) as complementary neuroimaging predictors. VMAP 
participants underwent brain magnetic resonance imaging at the 
Vanderbilt University Institute of Imaging Science using a 3T Philips 
Achieva system (Best, the Netherlands). Regional GM volumes were 
derived from T1-weighted magnetic resonance images, which were 
skull-stripped using SynthStrip \cite{hoopes2022synthstrip} and 
parcellated into anatomical regions of interest (ROIs) using 
NiChart\_DLMUSE~\cite{bashyam2025dlmuse}. Regional CBF estimates were 
derived from pseudo-continuous arterial spin-labeling (pCASL) images 
and corrected for partial-volume effects. Cerebrospinal fluid (CSF) 
$A\beta_{1-42}$ was measured using the Fujirebio INNOTEST assay, and 
was additionally considered as a biomarker associated with cerebral 
amyloid pathology~\cite{jack2018nia}. Detailed protocols for 
neuroimaging acquisition and processing and CSF collection are 
provided in Appendix~\ref{app:data}.

VMAP participants completed a comprehensive neuropsychological 
assessment covering episodic memory, executive function, language 
and information processing speed. The 
present analyses used VMAP composite scores for episodic memory and 
executive function. These standardized scores were constructed from 
item-level data using bifactor latent-variable models in which each 
measure loaded on both a general domain factor and a test-specific 
factor~\cite{mukherjee2023cognitive}. The episodic memory composite 
incorporated learning, recall, and recognition measures from the 
California Verbal Learning Test, Second Edition (CVLT-II), and the 
Biber Figure Learning Test (BFLT). The executive function composite 
incorporated the Delis--Kaplan Executive Function System (D-KEFS) 
Color--Word Inhibition, Tower, and Letter--Number Switching tests, 
together with Letter Fluency (FAS). 

\subsection{Analysis Procedure}
\label{sec:analysis}

We conducted two multivariate analyses. The first jointly 
considered memory and executive function, two related but distinct 
cognitive domains, to identify neuroimaging associations that were 
shared across or specific to each domain. The second jointly 
considered memory and CSF amyloid $A\beta_{1-42}$. This analysis 
had two objectives: to examine how the neuroimaging predictors were 
associated with cognitive function and AD pathology, and to 
investigate the downstream association between memory and CSF 
$A\beta_{1-42}$. Specific for the second analysis, because CSF 
$A\beta_{1-42}$ was unavailable for $53.1\%$ of the participants, 
SHIM was used to generate posterior-based multiple imputations of 
the missing data. The completed datasets were then used to estimate 
the association between memory and the CSF biomaker while accounting 
for uncertainty arising from the missing data.

The predictor hierarchy consisted of two neuroimaging modalities: 
regional GM and CBF, each containing $98$ predictors. Within each 
modality, predictors were further organized into five anatomical 
groups, including frontal, limbic, occipital, parietal, and 
temporal. This hierarchy was supplied to SHIM to guide hierarchical 
shrinkage and variable selection. All neuroimaging predictors were 
standardized prior to model implementation.

All analyses adjusted for age, sex, {\em APOE-$\varepsilon4$} 
carrier status, race and ethnicity, and years of education. 
Intracranial volume (ICV) was also included as a covariate in 
analyses to account for individual differences in head size. All 
adjusting covariates were included in SHIM as unpenalized variables. 
Participants in the analytic sample (of $288$) did not have missing 
data in the neuroimaging predictors. 

In the first analysis, we applied SHIM and its univariate 
counterpart, SHIM\_U, to episodic memory and executive function. 
SHIM\_U was included because it was among the most competitive 
competing alternative in the simulation study and accommodated the 
same unpenalized adjusting covariates as SHIM. MBSGSSS was not 
included in the analysis since it was not able to directly 
accommodate unpenalized covariates.

In the second analysis, we applied posterior-based multiple 
imputation by SHIM, complete-case analysis, and multiple imputation 
by chained equations (MICE)~\cite{vanbuuren2011mice}. For MICE, the 
default method was used, and for each imputation, $20$ completed 
datasets were generated. The imputation model included memory, 
neuroimaging predictors, and all adjusting covariates by 
convention~\cite{little2019statistical}. Within each completed 
dataset, memory was regressed on standardized CSF $A\beta_{1-42}$ 
with adjustment for the demographic and clinical covariates 
described above. The resulting estimates and variances were combined 
using Rubin's rules~\cite{rubin1987multiple}. 

\subsection{Results}
\label{sec:res}

In the first analysis, SHIM identified positive associations between 
GM in the left parahippocampal gyrus and episodic memory 
($\hat{\beta} = 0.243$, $95\%$ CI [credible interval]: 
$0.038$--$0.410$), and between GM in the left inferior frontal gyrus, 
pars triangularis and executive function ($\hat{\beta} = 
0.163$, $95\%$ CI: $0.011$--$0.277$). SHIM\_U identified only the 
association between left parahippocampal GM and episodic memory 
($\hat{\beta} = 0.254$, $95\%$ CI: $0.007$--$0.421$). No regional CBF 
measure met the selection criterion under either approach. Overall, 
SHIM selected two of the $392$ possible associations between the 
$196$ regional imaging measures and the two cognitive outcomes, both 
involving GM volume. This parsimonious result suggests that, after 
accounting jointly for the correlated imaging measures and 
covariates, posterior evidence was concentrated in a small number of 
anatomically specific regions. The selected regions were broadly 
consistent with previous neuroimaging evidence. Parahippocampal gray 
matter has been implicated in episodic memory in older adults, and 
reduced parahippocampal volume has been observed among individuals 
with memory impairment~\cite{echavarri2011atrophy, 
kohncke2021hippocampal}. Similarly, previous studies have linked 
inferior frontal GM volume to executive-function 
performance~\cite{lee2024baseline,perez2020subtle}. Notably, the 
executive-function association was identified under multivariate SHIM 
but not under SHIM\_U, underscoring the potential benefit of 
borrowing information across related cognitive outcomes in 
multivariate models.

In the second analysis, the estimated downstream association between 
episodic memory and CSF $A\beta_{1-42}$ is summarized in 
Table~\ref{tab:app2}. The pooled analysis based on SHIM posterior 
imputations indicated a positive association, with higher CSF 
$A\beta_{1-42}$ associated with better episodic-memory 
performance ($\hat{\beta} = 0.178$, $95\%$ CI [confidence interval]: 
$0.044$--$0.313$, p$=0.01$). This relationship is consistent with the 
established AD biomarker literature, where lower 
CSF $A\beta_{1-42}$ reflects greater cerebral amyloid 
deposition and is associated with poorer cognitive outcomes 
\cite{hansson2006association, li2014cross, rami2011cerebrospinal}.
The complete case analysis produced a similar point estimate but a 
greater standard error and a wider confidence interval, resulting in 
statistical insignificance at the nominal level. The lower precision 
indicated the loss of information from excluding the $153$ 
participants without observed CSF measurements. MICE produced a 
counter-intuitive result. With $196$ correlated imaging predictors 
but only $135$ observed CSF measurements, MICE was potentially 
sensitive to collinearity and imputation model specification. In 
contrast, SHIM jointly modeled the outcomes while applying 
hierarchical shrinkage to the imaging predictors and 
generating posterior imputations of the missing CSF measurements. In 
this application, the integrated approach produced an association 
estimate consistent with prior biological evidence and greater 
precision than the complete case analysis.

\begin{table}[tbh]
	\centering
	\caption{Estimated association between episodic memory and CSF 
	$A\beta_{1-42}$ using posterior-based multiple imputation by SHIM,
	complete case analysis, and MICE. CI: confidence interval}
	\label{tab:app2}
	\setlength{\tabcolsep}{6.9pt}
	\begin{tabular}{lcccc}
		\toprule
		Method & Sample Size & Estimate & 95\% CI & P-value \\
		\midrule
		SHIM & $288$ & $0.178$ & $(0.044, 0.313)$ & $0.010$ \\
		Complete case & $135$ & $0.178$ & $(-0.001, 0.357)$ & $0.054$ 
		\\
		MICE & $288$ & $-0.006$ & $(-0.014, 0.001)$ & $0.099$ \\
		\bottomrule
	\end{tabular}
\end{table}

\section{Discussion}
\label{sec:dis}

In this paper, we proposed SHIM, a Bayesian framework for structured 
variable selection in studies involving high-dimensional, multi-modal 
predictors and partially observed multivariate psychometric outcomes. 
Such settings increasingly arise when multiple related behavioral or 
cognitive constructs are studied alongside neuroimaging, genetic, or 
other high-dimensional auxiliary information. SHIM combines 
hierarchical horseshoe shrinkage with Bayesian treatment of the 
missing data, allowing the predictor hierarchy, dependence among the 
outcomes, and incomplete outcome data to be accommodated within a 
unified framework. The theoretical results formalize two properties 
of the proposed prior. Coefficients that are closer in the predictor 
hierarchy have more strongly associated prior magnitudes, and the 
marginal coefficient density combines substantial concentration near 
zero with heavy tails. Across the primary simulations and sensitivity 
analyses, SHIM generally achieved a balance between sensitivity and 
false positive control while providing accurate coefficient estimates 
and reasonably calibrated uncertainty quantification, although its 
performance deteriorated when the MAR assumption was violated. The 
VMAP application further illustrated two practical uses of the 
framework: borrowing information across related cognitive outcomes to 
identify neuroimaging associations and generating posterior-based 
multiple imputations for a downstream analysis involving an 
incompletely observed CSF biomarker.

The present study has several limitations that warrant further 
investigation. First, although SHIM accommodates missing outcomes, it 
currently requires the high-dimensional predictors and adjusting 
covariates to be fully observed, which may be restrictive in 
practice. Second, inference relies on the MAR assumption. The 
increased bias and under-coverage observed in the MNAR sensitivity 
setting indicate that performance may deteriorate when this 
assumption is severely violated. Third, the current formulation is 
primarily based on continuous outcomes modeled using a multivariate 
normal distribution. Extensions to binary, ordinal and mixed-type 
outcomes would broaden its applicability to psychometric research. 
Finally, SHIM requires the predictor hierarchy to be specified in 
advance, and its performance and interpretation may depend on how 
well that hierarchy reflects the underlying scientific structure. 
Future methodological work could therefore address incomplete 
predictors, non-Gaussian or latent psychometric outcomes, 
nonignorable missingness mechanisms, and more flexible or 
data-adaptive predictor hierarchies.

Establishing high-dimensional posterior theory for SHIM also remains 
an important direction for future research. The shared modality-, 
group-, and predictor-level scales induce dependence across the 
coefficient matrix, requiring joint prior-concentration arguments for 
hierarchically structured sparse signals rather than coordinate-wise 
calculations. The residual covariance matrix must also be learned 
jointly with the regression surface, while the observed-data 
likelihood under MAR varies across missingness patterns. A possible 
strategy is to combine general prior-concentration and testing 
frameworks for posterior contraction~\cite{ghosal2007convergence} 
with techniques developed for high-dimensional regression under 
continuous shrinkage priors~\cite{song2023nearly} and multivariate 
regression with an unknown covariance matrix~\cite{zhang2022ultra}. 
Establishing variable-selection consistency and valid uncertainty 
quantification would additionally require appropriate sparsity and 
signal-strength conditions~\cite{pas2017uncertainty, song2023nearly}. 
Once joint posterior contraction for the regression and covariance 
parameters has been established under the observed-data likelihood, 
corresponding contraction of the posterior predictive imputation 
distribution may follow from continuity of the conditional 
multivariate normal distribution. These developments would provide a 
fuller theoretical foundation for structured Bayesian variable 
selection and imputation in high-dimensional psychometric studies.

\begin{Backmatter}



\paragraph{Funding Statement}
This research was supported by grants from the NIH (R01AG034962).

\paragraph{Competing Interests}
T.J.H.\ serves on the Scientific Advisory Board of Circular Genomics, 
as Deputy Editor of \textit{Alzheimer's \& Dementia: Translational 
Research \& Clinical Interventions}, and as Section Editor of 
\textit{Alzheimer's \& Dementia}. A.L.J\ is an advisor for Medtronic, 
and Chair of the Observational Study Monitoring Board for the 
Diverse-VCID: White Matter Lesion Etiology of Dementia in Diverse 
Populations Study.

\paragraph{Data Availability Statement}
The VMAP data analyzed in this study are not publicly available. 
Qualified investigators may apply for access through the VMAP Data 
Sharing Portal at \url{https://vmacdata.org/}. Requests are subject 
to review and approval by the VMAP scientific committee and require 
completion of a data-use agreement. Code used to implement the 
analyses reported in this article is available at 
\url{https://github.com/zongyue-teng/shim}.

\paragraph{Ethical Standards}
The research meets all ethical guidelines, including adherence to the legal requirements of the study country.

\paragraph{Author Contributions}
Conceptualization: Z.T.; S.M.; T.J.H.; P.Z. Data curation: T.J.H.; 
A.L.J. Formal Analysis: Z.T. Funding Acquisition: A.L.J. Methodology: 
Z.T; S.M; P.Z. Software: Z.T.; P.Z. Supervision: P.Z. Validation: 
Z.T.; P.Z. Visualization: Z.T. Writing original draft: Z.T.; P.Z. 
Writing review \& editing: T.Z.; S.M.; T.J.H.; A.L.J.; P.Z. All 
authors approved the final submitted draft.

\bibliography{refs}

\end{Backmatter}

\begin{appendix}
	\appendix
	
	\section{Derivation of Gaussian Full Conditional Distributions}
	\label{app:conditionals}
	
	This appendix provides derivations of the Gaussian full 
	conditional	distributions for the regression coefficients and 
	intercept vector used in the Gibbs sampler, i.e., 
	Algorithm~\ref{alg:gibbs}. The remaining updates either follow
	directly from standard conjugate Bayesian analyses (e.g., the
	inverse-Wishart update for the residual covariance matrix) or 
	from the auxiliar-variable representation of the half-Cauchy 
	prior~\cite{makalic2016simple}, and are therefore omitted. 
	Throughout this appendix, ``others'' denotes conditioning on the 
	current values of all remaining unknown quantities.
	
	\subsection{Regression Coefficients}
	\label{app:conditionals_B}	
	Denote
	\[
	\bm{w} = \mathrm{vec} \left\{(\bm{Y} - 
	\bm{1}_n\bm{\alpha}^{\top} - \bm{Z}\bm{\Psi})^{\top}\right\} \qquad \text{and} 
	\qquad \bm{b} = \mathrm{vec}(\bm{B}^{\top}).
	\]
	Using the ``vec trick'' identity~\cite{harville1997kronecker} 
	given by
	\[
	\mathrm{vec}(\bm{M}_1 \bm{M}_2 \bm{M}_3) = (\bm{M}_3^{\top} 
	\otimes \bm{M}_1) \mathrm{vec}(\bm{M}_2),
	\]
	where $\otimes$ denotes the Kronecker product, we re-express the 
	$\bm{w}$ as
	\[\bm{w} = \mathrm{vec} \left\{(\bm{Y} - 
	\bm{1}_n\bm{\alpha}^{\top} - \bm{Z}\bm{\Psi})^{\top}\right\} = 
	\mathrm{vec}(\bm{B}^{\top} \bm{X}^{\top}) + 
	\mathrm{vec}(\bm{E}^{\top}) \\
	= (\bm{X} \otimes \bm{I}_Q) \bm{b} + \mathrm{vec}(\bm{E}^{\top}),
	\]
	which suggests $\bm{w} \mid \bm{b}, \bm{\Sigma} \sim
	\mathcal{N}_{nQ} (\bm{H}\bm{b},\bm{R})$
	with $\bm{H} := \bm{X} \otimes \bm{I}_Q$ and $\bm{R}
	:= \bm{I}_n \otimes \bm{\Sigma}$.
	
	Since the prior distribution of $\bm{b}$ is specified as $\bm{b}
	\sim \mathcal{N}_{PQ} (\bm{0},\bm{D}_{\bm{B}})$, the posterior 
	density is given by
	\begin{align*}
		p(\bm{b} \mid \mathrm{others}) &\propto p(\bm{w} \mid \bm{b}, 
		\bm{\Sigma}) p(\bm{b}) \\
		&\propto \exp \left\{-\frac{1}{2} (\bm{w} - 
		\bm{H}\bm{b})^{\top} \bm{R}^{-1} (\bm{w} - \bm{H}\bm{b})
		- \frac{1}{2} \bm{b}^{\top} \bm{D}_{\bm{B}}^{-1} 
		\bm{b}\right\} \\
		&\propto \exp \left\{-\frac{1}{2} \bm{b}^{\top} \left(
		\bm{H}^{\top} \bm{R}^{-1} \bm{H} + \bm{D}_{\bm{B}}^{-1}
		\right)	\bm{b} + \bm{b}^{\top} \bm{H}^{\top} \bm{R}^{-1}
		\bm{w} \right\},
	\end{align*}
	which leads to 
	\[
	\bm{b} \mid \mathrm{others} \sim
	\mathcal{N}_{PQ}(\bm{\mu}_{\bm{B}}, \bm{\Sigma}_{\bm{B}}),
	\]
	with
	\[
	\bm{\Sigma}_{\bm{B}} = \left(\bm{H}^{\top} \bm{R}^{-1}
	\bm{H} + \bm{D}_{\bm{B}}^{-1} \right)^{-1} \qquad \text{and}
	\qquad \bm{\mu}_{\bm{B}} = \bm{\Sigma}_{\bm{B}}
	\bm{H}^{\top} \bm{R}^{-1} \bm{w}
	\]
	according to the multivariate normal distribution kernel. 
	Finally, plugging $\bm{H}$ and $\bm{R}$ in the expressions of 
	$\bm{\Sigma}_{\bm{B}}$ and $\bm{\mu}_{\bm{B}}$ to get
	\[
	\bm{\Sigma}_{\bm{B}} = \left((\bm{X} \otimes 
	\bm{I}_Q)^{\top} (\bm{I}_n \otimes \bm{\Sigma})^{-1} (\bm{X} 
	\otimes \bm{I}_Q) + \bm{D}_{\bm{B}}^{-1}\right)^{-1} \\
	= \left((\bm{X}^{\top} \bm{X}) \otimes \bm{\Sigma}^{-1} + 
	\bm{D}_{\bm{B}}^{-1}\right)^{-1} 
	\]
	and
	\[
	\bm{\mu}_{\bm{B}} = \bm{\Sigma}_{\bm{B}} (\bm{X} \otimes 
	\bm{I}_Q)^{\top} (\bm{I}_n \otimes \bm{\Sigma})^{-1} \bm{w} = 
	\bm{\Sigma}_{\bm{B}} (\bm{X}^{\top} \otimes 
	\bm{\Sigma}^{-1}) \bm{w},
	\]	
	which completes the derivations.
	
	\subsection{Intercept Vector and Covariate Coefficients}
	\label{app:conditionals_alpha}
	
	To derive the joint full conditional distribution of the intercept 
	and covariate coefficient matrix $\widetilde{\bm{\Psi}}$, consider
	\[
	\bm{u} = \mathrm{vec} \left\{(\bm{Y} - \bm{X}\bm{B})^{\top} \right\} \qquad \text{and} 
	\qquad \tilde{\bm{\psi}} = \mathrm{vec}(\widetilde{\bm{\Psi}}^{\top}).
	\]
	Similarly, we can re-express $\bm{u}$ as
	\[\bm{u} =
	\mathrm{vec}(\widetilde{\bm{\Psi}}^{\top} \widetilde{\bm{Z}}^{\top}) + 
	\mathrm{vec}(\bm{E}^{\top}) \\
	= (\widetilde{\bm{Z}} \otimes \bm{I}_Q) \tilde{\bm{\psi}} + \mathrm{vec}(\bm{E}^{\top}),
	\]
	which suggests $\bm{u} \mid \tilde{\bm{\psi}}, \bm{\Sigma} \sim
	\mathcal{N}_{nQ} (\bm{G}\tilde{\bm{\psi}},\bm{R})$
	with $\bm{G} := \widetilde{\bm{Z}} \otimes \bm{I}_Q$ and $\bm{R}
	= \bm{I}_n \otimes \bm{\Sigma}$.

	Given the specified prior for $\widetilde{\bm{\Psi}}$, denote 
	$\tilde{\bm{\psi}}_0 = \mathrm{vec}(\widetilde{\bm{\Psi}}_0^{\top})$ and 
	$\bm{D}_{\widetilde{\bm{\Psi}}} = \widetilde{\bm{V}}_0 \otimes \bm{I}_Q$ for notation
	simplicity. The posterior density is given by	
	\begin{align*}
		p(\tilde{\bm{\psi}} \mid \mathrm{others}) &\propto p(\bm{u} \mid \tilde{\bm{\psi}}, 
		\bm{\Sigma}) p(\tilde{\bm{\psi}}) \\
		&\propto \exp \left\{-\frac{1}{2} (\bm{u} - 
		\bm{G}\tilde{\bm{\psi}})^{\top} \bm{R}^{-1} (\bm{u} - \bm{G}\tilde{\bm{\psi}})
		- \frac{1}{2} (\tilde{\bm{\psi}} - \tilde{\bm{\psi}}_0)^{\top} 
		\bm{D}_{\widetilde{\bm{\Psi}}}^{-1} 
		(\tilde{\bm{\psi}} - \tilde{\bm{\psi}}_0)\right\} \\
		&\propto \exp \left\{-\frac{1}{2} \tilde{\bm{\psi}}^{\top} \left(
		\bm{G}^{\top} \bm{R}^{-1} \bm{G} + \bm{D}_{\widetilde{\bm{\Psi}}}^{-1}
		\right)	\tilde{\bm{\psi}}+ \tilde{\bm{\psi}}^{\top} \left(\bm{G}^{\top} \bm{R}^{-1} \bm{u} + \bm{D}_{\widetilde{\bm{\Psi}}}^{-1} \tilde{\bm{\psi}}_0 \right) \right\},
	\end{align*}	
	which corresponds to 
	\[
	\tilde{\bm{\psi}} \mid \mathrm{others} \sim
	\mathcal{N}_{(L+1)Q}(\bm{\mu}_{\widetilde{\bm{\Psi}}}, \bm{\Sigma}_{\widetilde{\bm{\Psi}}}),
	\]
	with
	\[
	\bm{\Sigma}_{\widetilde{\bm{\Psi}}} =  \left(
		\bm{G}^{\top} \bm{R}^{-1} \bm{G} + \bm{D}_{\widetilde{\bm{\Psi}}}^{-1}
		\right)^{-1} \qquad \text{and}
	\qquad \bm{\mu}_{\widetilde{\bm{\Psi}}} = \bm{\Sigma}_{\widetilde{\bm{\Psi}}}
	\left(\bm{G}^{\top} \bm{R}^{-1} \bm{u} + \bm{D}_{\widetilde{\bm{\Psi}}}^{-1} \tilde{\bm{\psi}}_0 \right).
	\]
	Plugging in for $\bm{G}$, $\bm{R}$, $\bm{D}_{\widetilde{\bm{\Psi}}}$, $\bm{u}$ and $\tilde{\bm{\psi}}_0$ leads to
	\[
	\bm{\Sigma}_{\widetilde{\bm{\Psi}}}= \left((\widetilde{\bm{Z}} \otimes 
	\bm{I}_Q)^{\top} (\bm{I}_n \otimes \bm{\Sigma})^{-1} (\widetilde{\bm{Z}}
	\otimes \bm{I}_Q) + ( \widetilde{\bm{V}}_0 \otimes \bm{I}_Q)^{-1}\right)^{-1} \\
	= \left((\widetilde{\bm{Z}}^{\top} \widetilde{\bm{Z}}) \otimes \bm{\Sigma}^{-1} + 
	\widetilde{\bm{V}}_0^{-1} \otimes \bm{I}_Q\right)^{-1} 
	\]
	and
	\[
	\bm{\mu}_{\widetilde{\bm{\Psi}}} = 
	\bm{\Sigma}_{\widetilde{\bm{\Psi}}} 
	\left\{(\widetilde{\bm{Z}}^\top\otimes\bm{\Sigma}^{-1})\mathrm{vec}
	\left((\bm{Y} - \bm{X}\bm{B})^\top\right) + (\widetilde{\bm{V}}_0^{-1}\otimes\bm{I}_Q) \mathrm{vec}
	(\widetilde{\bm{\Psi}}_0^\top) \right\}.
	\]	
	This completes the derivations.

	\section{Proof of Theorem~\ref{thm:marginal_prior}}
	\label{app:marginal_prior}
	
	\begin{proof}
		For $x > 0$, define the one-sided Mellin transform of the 
		symmetric density $\pi_K$ by
		\[
		\mathcal{M}_K(s) = \int_0^\infty x^{s - 1} \pi_K(x) 
		\,\mathrm{d}x.
		\]
		By symmetry, we get
		\[
		\mathcal M_K(s)	= \frac{1}{2}\operatorname{E}\left(|\beta|^{s 
		- 1}\right).
		\]
		Write
		\[
		\beta = S_KZ, \qquad S_K = \prod_{\ell = 1}^K H_\ell,
		\]
		where $Z \sim \mathcal{N}(0, 1)$ and $H_1, \ldots, H_K 
		\overset{\mathrm{ind}}{\sim} \mathcal{C}^+(0, 1)$
		are mutually independent. It follows that
		\[
		\mathcal{M}_K(s) = \frac{1}{2} \operatorname{E}\left(|Z|^{s - 
		1}\right)\prod_{\ell = 1}^K	\operatorname{E}\left(H_\ell^{s 
		- 1}\right).
		\]
		
		For a standard half-Cauchy variable, we know
		\[
		\operatorname{E}\left(H_\ell^{s - 1}\right)
		= \frac{2}{\pi}\int_0^\infty\frac{h^{s - 1}}{1 + h^2} 
		\,\mathrm{D}h =	\csc\left(\frac{\pi s}{2}\right)
		\]
		with $0 < \operatorname{Re}(s) < 2$.
		For $Z \sim \mathcal{N}(0, 1)$, we have
		\[
		\operatorname{E}\left(|Z|^{s-1}\right) = \frac{2^{(s - 
		1)/2}\Gamma(s/2)}{\sqrt{\pi}}
		\]
		with
		$\operatorname{Re}(s) > 0$. Therefore, we get
		\begin{equation}
			\label{eq:mellin_pi}
			\mathcal{M}_K(s) =	\frac{2^{(s - 1)/2}\Gamma(s/2)}
			{2\sqrt{\pi}}\left\{\csc\left(\frac{\pi s}{2}\right)
			\right\}^K
		\end{equation}
		for $0 < \operatorname{Re}(s) < 2$, which is the fundamental 
		strip of $\mathcal{M}_K$. For convenience, define
		\[
		A_K	= \frac{1}{\sqrt{2\pi}}\left(\frac{2}{\pi}\right)^K.
		\]
		
		We first consider the behavior near zero. As $s \rightarrow 
		0$, we have
		\begin{align*}
			2^{(s - 1)/2} &= 2^{-1/2}\{1 + O(s)\}, \\
			\Gamma(s/2) &= \frac{2}{s}\{1 + O(s)\}, \\
			\csc\left(\frac{\pi s}{2}\right) &=	\frac{2}{\pi s}\{1 + 
			O(s^2)\}.
		\end{align*}
		Substitution into Equation~\eqref{eq:mellin_pi} gives
		\[
		\mathcal M_K(s)	= A_K s^{-(K + 1)}\{1 + O(s)\}.
		\]
		Thus, $\mathcal{M}_K$ has a pole of order $K + 1$ at the left
		boundary $s = 0$. The converse mapping theorem for Mellin 
		transforms \cite{flajolet1995mellin} therefore gives
		\[
		\pi_K(x) \sim 
		\frac{A_K}{K!}\left\{\log\left(\frac{1}{x}\right)\right\}^K,
		\]
		as $\qquad x \downarrow0$.
		
		We next consider the tail behavior. Set $w = 2 - s$. As
		$w \rightarrow 0$, we have
		\begin{align*}
			2^{(s - 1)/2} &= 2^{(1 - w)/2} = \sqrt{2}\{1 + O(w)\},\\
			\Gamma(s/2) &= \Gamma(1 - w/2) = 1 + O(w), \\
			\csc\left(\frac{\pi s}{2}\right) &= \csc\left(\frac{\pi 
			w}{2}\right) = \frac{2}{\pi w}\{1 + O(w^2)\}.
		\end{align*}
		Hence, we get
		\[
		\mathcal{M}_K(2 - w) = A_K w^{-K}\{1 + O(w)\}.
		\]
		Equivalently, $\mathcal{M}_K(s)$ has a pole of order $K$ at 
		the right boundary $s=2$. A second application of the 
		converse mapping theorem yields
		\[
		\pi_K(x) \sim \frac{A_K}{(K - 1)!}\frac{(\log x)^{K - 
		1}}{x^2},
		\]
		as $x \rightarrow \infty$. Symmetry of $\pi_K$ allows $x$ to 
		be replaced by $|\beta|$.
		
		It remains to verify the analytic conditions required by the
		converse mapping theorem. The expression in
		Equation~\eqref{eq:mellin_pi} extends meromorphically beyond 
		its
		fundamental strip. Moreover, for fixed real $\sigma$ on a 
		vertical
		line avoiding its poles, Stirling's formula and the standard
		asymptotics of the cosecant function give
		\[
		\left|\Gamma\left(\frac{\sigma+it}{2}\right)\right|	=
		O\left((1+|t|)^{\sigma/2 - 1/2} e^{-\pi|t|/4}
		\right)
		\qquad \text{and}
		\qquad
		\left|
		\csc\left(\frac{\pi(\sigma+it)}{2}\right)
		\right|	= O\left(e^{-\pi|t|/2}\right)
		\]
		as $|t| \rightarrow \infty$. Thus, $\mathcal{M}_K(\sigma+it)$
		decays exponentially on such vertical lines, which completes 
		the proof.
	\end{proof}

	\section{Example Trace Plots for Bayesian Methods}
	\label{app:trace}
	
	Figure~\ref{fig:trace} presents representative trace plots and 
	for all four Bayesian methods under one benchmark simulation 
	setting. Similar convergence behavior was observed across the 
	remaining simulation settings.
	
	\begin{figure}[tbp]
		\centering
   		 \caption{Example trace plots of the four active regression 
   		 coefficients for the four Bayesian methods under Scenario 
   		 S6, based on $2{,}000$ MCMC iterations with the first $500$ 
   		 iterations discarded as burn-in}
    	\includegraphics[width=\textwidth]{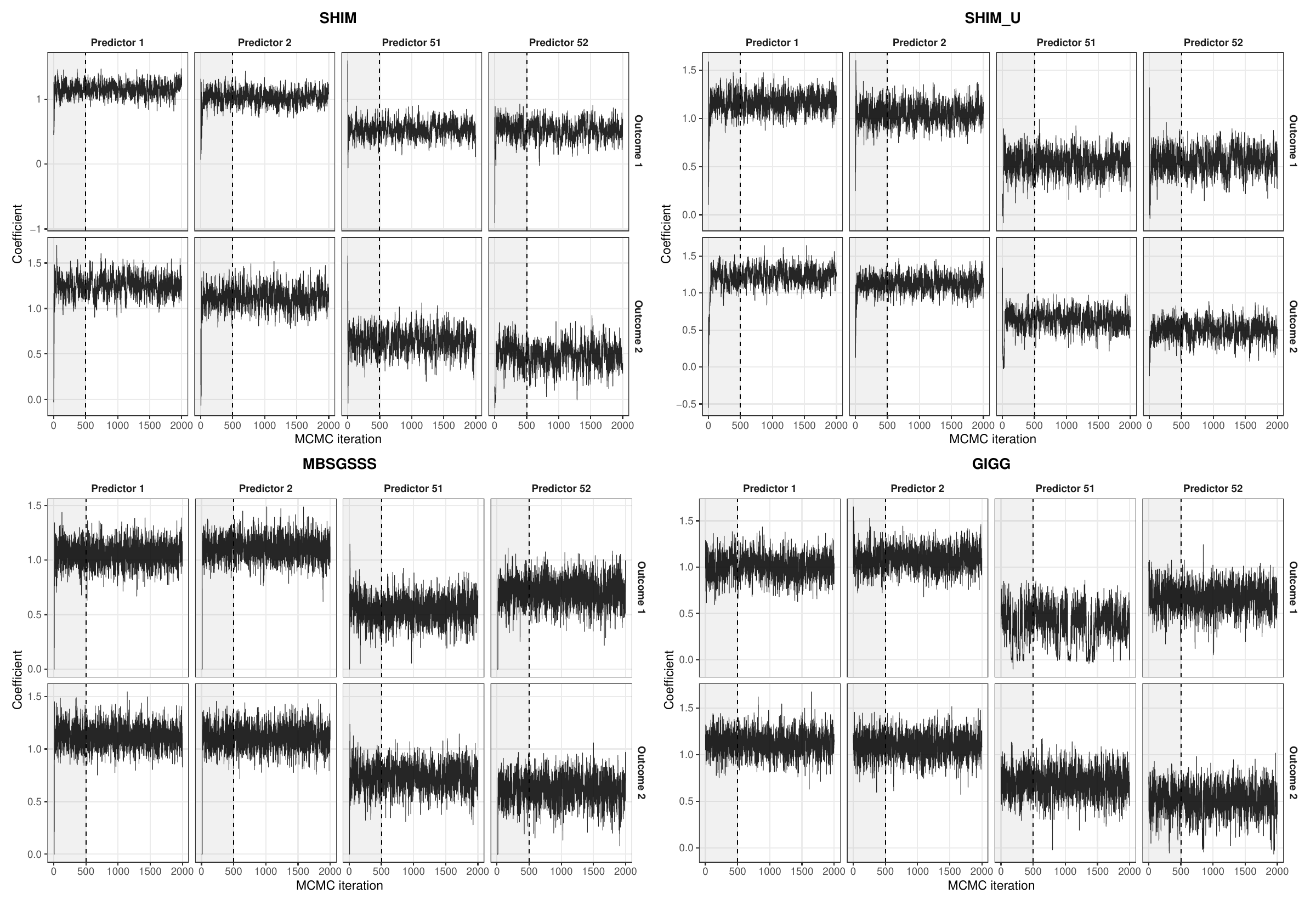}
    	\label{fig:trace}
	\end{figure}
	
	\section{Additional Estimation Results}
	\label{app:add_sd}
	
	For completeness, this appendix reports two additional estimation
	results that were not presented in the main text. 
	Figure~\ref{fig:sd} shows the empirical standard deviation of the 
	coefficient estimates and the ratio of the empirical standard 
	deviation to the average estimated standard error. Overall, SHIM, 
	SHIM\_U, and MBSGSSS exhibited comparable empirical	standard 
	deviations, with slightly larger variability observed for the
	weaker coefficients under the more challenging simulation 
	settings. GIGG demonstrated obviously greater variability for the 
	weaker signals, whereas glasso generally produced the largest 
	empirical standard deviations, particularly when both predictor 
	dimensionality and missingness were high. The 
	empirical-to-estimated standard error ratios for SHIM remained 
	close to one across nearly all simulation settings, indicating 
	good agreement between empirical sampling variability and 
	model-based uncertainty estimates. Similar behavior
	was observed for SHIM\_U and MBSGSSS, whereas larger deviations 
	from one were found for lasso, glasso, and GIGG in the more 
	challenging settings.
	
	\begin{figure}[tbp]
	\centering
    	\caption{Empirical standard deviation (top) and 
    	empirical-to-estimated standard error (bottom) for the four 
    	active coefficients across the eight simulation settings. 
    	Coefficients $\beta_1$ and $\beta_2$ correspond to the two 
    	stronger signals (true coefficient $ = 1.2$), whereas 
    	$\beta_3$ and $\beta_4$ correspond to the two weaker signals 
    	(true coefficient $ = 0.6$).  For lasso and group lasso, the reported summaries 
    	are based on post-selection ordinary least squares estimates}
    \includegraphics[width=\textwidth]{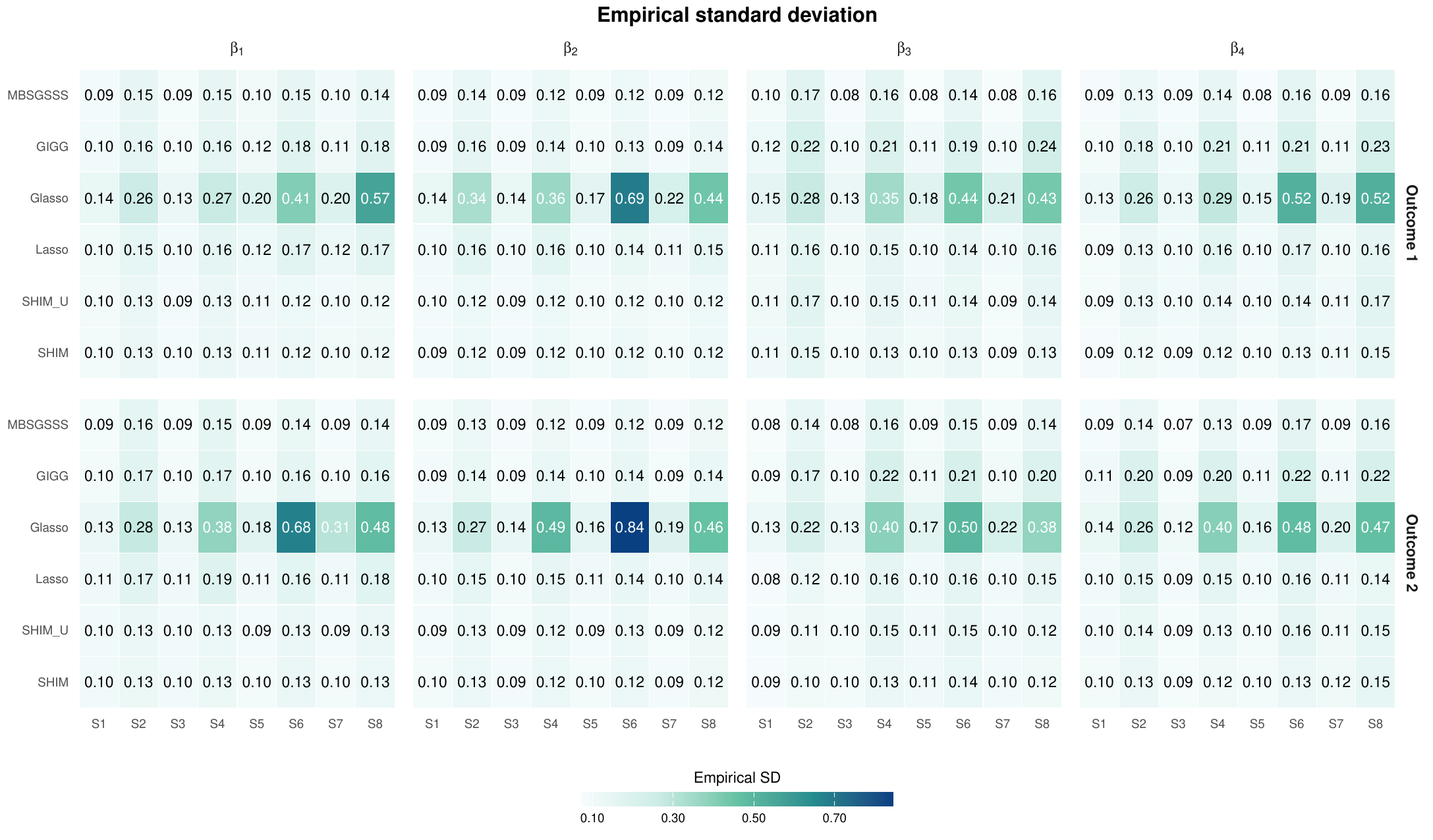}
    \includegraphics[width=\textwidth]{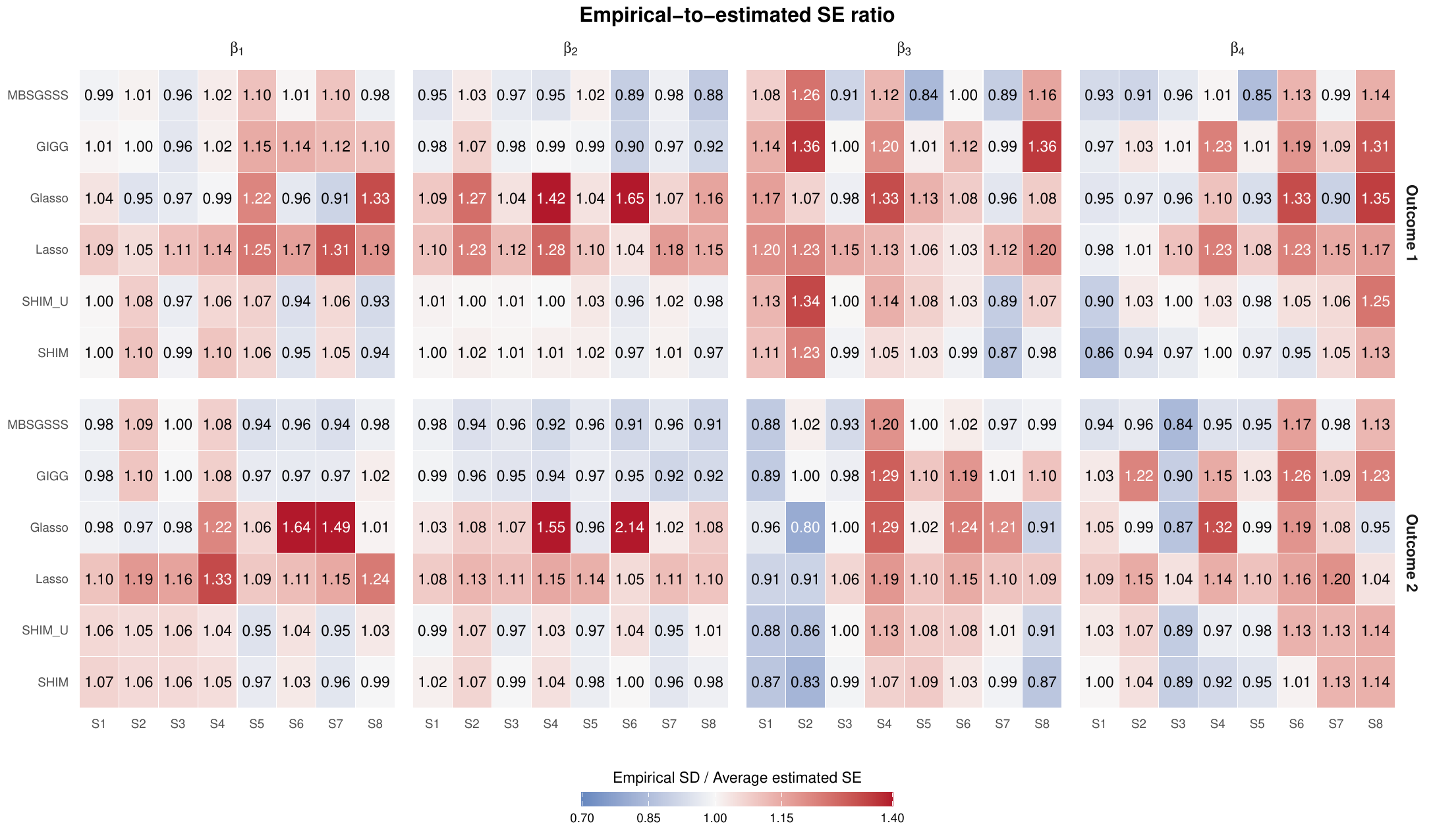}
    \label{fig:sd}
\end{figure}
	
\section{Full Results for Sensitivity Analysis}
\label{app:res_sens}

Table~\ref{tab:sel_sens} provides the full variable-selection results. As discussed in Section~\ref{sec:sens}, SHIM and SHIM\_U maintained high F1 scores across all sensitivity settings, with SHIM consistently showing higher TPR values than SHIM\_U. MBSGSSS was comparable to SHIM and SHIM\_U under most settings except OUT, where its F1 scores deteriorated substantially, primarily because of increased FDR. GIGG continued to exhibit lower TPR values, whereas lasso and glasso yielded inflated FDR values.

Table~\ref{tab:est_sens} provides the full 
coefficient-estimation results. Across settings, GIGG 
demonstrated greater bias and under-coverage for the weak signals, 
lasso showed persistent bias and under-coverage across the active 
coefficients, and glasso exhibited substantially greater empirical 
variability and poorer standard error calibration. In particular, 
for ULTRA, no valid post-selection OLS estimates were available for 
glasso because the number of selected predictors was consistently at 
least as large as the complete-case sample size.

Under CORR, SHIM, SHIM\_U, and MBSGSSS produced comparably low bias for the strong signals but greater bias for some weak signals. However, SHIM and SHIM\_U exhibited lower empirical standard deviations than MBSGSSS and the other competing methods, particularly for the weak signals, with SHIM generally showing slightly smaller values than SHIM\_U. Their empirical-to-average estimated standard error ratios were generally close to one, and coverage was close to the nominal 95\% level. In contrast, MBSGSSS showed noticeably poorer standard error calibration and some under-coverage for the weak signals.

Similar patterns were observed under MNAR, although bias increased 
and coverage deteriorated for most methods. Nevertheless, SHIM and 
SHIM\_U retained competitive bias and relatively low empirical 
variability. In contrast, MBSGSSS exhibited greater bias even for 
the strong signals, together with more pronounced standard error 
miscalibration and under-coverage for the weak signals. 

Under ULTRA, SHIM, SHIM\_U, and MBSGSSS performed comparably in 
terms of bias and empirical variability, with low bias for the 
strong signals but greater bias for some weak signals. All three 
methods showed some standard error miscalibration and under-coverage 
for the weak signals. MBSGSSS generally exhibited the greatest 
miscalibration, whereas SHIM remained competitive in calibration and 
coverage.

\begin{table}[tbp]
	\centering
	\scriptsize
	\caption{Average variable-selection performance over $100$
	simulation replicates. Results are reported separately for each
	outcome under the four sensitivity analysis settings. Performance
	is evaluated using the true positive rate (TPR), false positive rate
	(FPR), false discovery rate (FDR), and F1 score.}
	\label{tab:sel_sens}
	\setlength{\tabcolsep}{3.37pt}
	\begin{tabular}{lcccccccccccccccc}
	\toprule
	& \multicolumn{8}{c}{CORR}
	& \multicolumn{8}{c}{MNAR} \\
	\cmidrule(lr){2-9}
	\cmidrule(lr){10-17}
	& \multicolumn{4}{c}{Outcome 1}
	& \multicolumn{4}{c}{Outcome 2}
	& \multicolumn{4}{c}{Outcome 1}
	& \multicolumn{4}{c}{Outcome 2} \\
	\cmidrule(lr){2-5}
	\cmidrule(lr){6-9}
	\cmidrule(lr){10-13}
	\cmidrule(lr){14-17}
	& TPR & FPR & FDR & F1
	& TPR & FPR & FDR & F1
	& TPR & FPR & FDR & F1
	& TPR & FPR & FDR & F1 \\
	\midrule
	SHIM
	& 0.98 & 0.00 & 0.01 & 0.98
	& 0.98 & 0.00 & 0.01 & 0.98
	& 0.94 & 0.00 & 0.01 & 0.96
	& 0.95 & 0.00 & 0.01 & 0.97 \\
	SHIM\_U
	& 0.94 & 0.00 & 0.00 & 0.96
	& 0.94 & 0.00 & 0.02 & 0.96
	& 0.92 & 0.00 & 0.03 & 0.94
	& 0.91 & 0.00 & 0.02 & 0.94 \\
	Lasso
	& 0.99 & 0.02 & 0.61 & 0.55
	& 1.00 & 0.02 & 0.62 & 0.54
	& 0.98 & 0.02 & 0.63 & 0.52
	& 0.98 & 0.02 & 0.65 & 0.50 \\
	Glasso
	& 1.00 & 0.35 & 0.97 & 0.06
	& 1.00 & 0.37 & 0.97 & 0.06
	& 1.00 & 0.35 & 0.97 & 0.06
	& 1.00 & 0.36 & 0.97 & 0.06 \\
	GIGG
	& 0.82 & 0.00 & 0.01 & 0.89
	& 0.82 & 0.00 & 0.01 & 0.89
	& 0.77 & 0.00 & 0.01 & 0.85
	& 0.77 & 0.00 & 0.01 & 0.86 \\
	MBSGSSS
	& 0.96 & 0.00 & 0.06 & 0.95
	& 0.96 & 0.00 & 0.06 & 0.95
	& 0.95 & 0.00 & 0.06 & 0.94
	& 0.95 & 0.00 & 0.06 & 0.94 \\

	\midrule
	& \multicolumn{8}{c}{OUT}
	& \multicolumn{8}{c}{ULTRA} \\
	\cmidrule(lr){2-9}
	\cmidrule(lr){10-17}
	& \multicolumn{4}{c}{Outcome 1}
	& \multicolumn{4}{c}{Outcome 2}
	& \multicolumn{4}{c}{Outcome 1}
	& \multicolumn{4}{c}{Outcome 2} \\
	\cmidrule(lr){2-5}
	\cmidrule(lr){6-9}
	\cmidrule(lr){10-13}
	\cmidrule(lr){14-17}
	& TPR & FPR & FDR & F1
	& TPR & FPR & FDR & F1
	& TPR & FPR & FDR & F1
	& TPR & FPR & FDR & F1 \\
	\midrule
	SHIM
	& 0.96 & 0.00 & 0.06 & 0.94
	& 0.96 & 0.00 & 0.04 & 0.95
	& 0.95 & 0.00 & 0.04 & 0.95
	& 0.96 & 0.00 & 0.04 & 0.95 \\
	SHIM\_U
	& 0.95 & 0.00 & 0.01 & 0.96
	& 0.94 & 0.00 & 0.01 & 0.95
	& 0.91 & 0.00 & 0.01 & 0.94
	& 0.92 & 0.00 & 0.01 & 0.95 \\
	Lasso
	& 1.00 & 0.01 & 0.48 & 0.65
	& 0.99 & 0.01 & 0.53 & 0.60
	& 0.99 & 0.02 & 0.65 & 0.50
	& 0.99 & 0.02 & 0.66 & 0.50 \\
	Glasso
	& 0.97 & 0.33 & 0.98 & 0.03
	& 0.98 & 0.37 & 0.99 & 0.03
	& 1.00 & 0.33 & 0.98 & 0.04
	& 1.00 & 0.32 & 0.98 & 0.04 \\
	GIGG
	& 0.84 & 0.00 & 0.00 & 0.89
	& 0.80 & 0.00 & 0.01 & 0.86
	& 0.80 & 0.00 & 0.01 & 0.87
	& 0.80 & 0.00 & 0.01 & 0.88 \\
	MBSGSSS
	& 0.95 & 0.01 & 0.51 & 0.64
	& 0.94 & 0.01 & 0.52 & 0.63
	& 0.98 & 0.00 & 0.04 & 0.96
	& 0.98 & 0.00 & 0.04 & 0.97 \\
	\bottomrule
	\end{tabular}
\end{table}

\begin{table}[tbp]
	\centering

	\caption{Estimation performance for the active predictors under
		CORR, MNAR and ULTRA}
	\label{tab:est_sens}
	\footnotesize
	\setlength{\tabcolsep}{3.3pt}

	\begin{tabular}{lcccccccccccccccc}
	\toprule

	\textbf{CORR}
	& \multicolumn{16}{c}{Outcome 1} \\
	\cmidrule(lr){2-17}
	& \multicolumn{4}{c}{PB (\%)}
	& \multicolumn{4}{c}{ESD}
	& \multicolumn{4}{c}{ESD / SE}
	& \multicolumn{4}{c}{CP (\%)} \\
	\cmidrule(lr){2-5}
	\cmidrule(lr){6-9}
	\cmidrule(lr){10-13}
	\cmidrule(lr){14-17}
	Method
	& $\beta_1$ & $\beta_2$ & $\beta_3$ & $\beta_4$
	& $\beta_1$ & $\beta_2$ & $\beta_3$ & $\beta_4$
	& $\beta_1$ & $\beta_2$ & $\beta_3$ & $\beta_4$
	& $\beta_1$ & $\beta_2$ & $\beta_3$ & $\beta_4$ \\
	\midrule

	SHIM
	& 3.0 & 1.9 & 5.5 & 14.9
	& 0.11 & 0.11 & 0.12 & 0.12
	& 0.92 & 0.97 & 0.95 & 0.93
	& 95 & 95 & 98 & 93 \\

	SHIM\_U
	& 2.7 & 2.0 & 6.2 & 16.7
	& 0.12 & 0.12 & 0.13 & 0.14
	& 0.96 & 0.96 & 1.00 & 1.04
	& 95 & 97 & 91 & 91 \\

	Lasso
	& 15.9 & 14.8 & 25.0 & 33.9
	& 0.18 & 0.14 & 0.14 & 0.16
	& 1.21 & 1.05 & 1.03 & 1.21
	& 75 & 74 & 79 & 67 \\

	Glasso
	& 1.5 & 16.6 & 12.8 & 28.6
	& 0.44 & 0.79 & 0.45 & 0.53
	& 1.06 & 1.89 & 1.12 & 1.37
	& 95 & 95 & 93 & 95 \\

	GIGG
	& 3.6 & 2.9 & 12.6 & 25.4
	& 0.18 & 0.14 & 0.20 & 0.21
	& 1.11 & 0.95 & 1.13 & 1.22
	& 91 & 95 & 90 & 86 \\

	MBSGSSS
	& 2.8 & 2.6 & 4.4 & 16.2
	& 0.15 & 0.13 & 0.19 & 0.21
	& 1.02 & 0.93 & 1.20 & 1.35
	& 94 & 93 & 95 & 91 \\

	\midrule

	&
	\multicolumn{16}{c}{Outcome 2} \\
	\cmidrule(lr){2-17}
	& \multicolumn{4}{c}{PB (\%)}
	& \multicolumn{4}{c}{ESD}
	& \multicolumn{4}{c}{ESD / SE}
	& \multicolumn{4}{c}{CP (\%)} \\
	\cmidrule(lr){2-5}
	\cmidrule(lr){6-9}
	\cmidrule(lr){10-13}
	\cmidrule(lr){14-17}
	Method
	& $\beta_1$ & $\beta_2$ & $\beta_3$ & $\beta_4$
	& $\beta_1$ & $\beta_2$ & $\beta_3$ & $\beta_4$
	& $\beta_1$ & $\beta_2$ & $\beta_3$ & $\beta_4$
	& $\beta_1$ & $\beta_2$ & $\beta_3$ & $\beta_4$ \\
	\midrule

	SHIM
	& 3.2 & 3.3 & 6.2 & 12.4
	& 0.12 & 0.11 & 0.13 & 0.12
	& 0.99 & 0.98 & 1.04 & 0.94
	& 93 & 93 & 94 & 93 \\

	SHIM\_U
	& 3.1 & 3.5 & 6.1 & 14.1
	& 0.13 & 0.12 & 0.14 & 0.15
	& 1.05 & 1.01 & 1.05 & 1.07
	& 94 & 93 & 94 & 89 \\

	Lasso
	& 15.5 & 17.0 & 26.2 & 31.4
	& 0.17 & 0.15 & 0.15 & 0.15
	& 1.16 & 1.08 & 1.07 & 1.14
	& 77 & 69 & 81 & 74 \\

	Glasso
	& 3.7 & 12.6 & 8.0 & 13.3
	& 0.66 & 0.91 & 0.52 & 0.53
	& 1.62 & 2.29 & 1.31 & 1.34
	& 86 & 95 & 92 & 95 \\

	GIGG
	& 3.0 & 4.6 & 12.9 & 22.6
	& 0.16 & 0.14 & 0.21 & 0.22
	& 1.01 & 0.95 & 1.15 & 1.29
	& 96 & 91 & 91 & 83 \\

	MBSGSSS
	& 2.6 & 3.6 & 5.2 & 13.9
	& 0.15 & 0.13 & 0.19 & 0.22
	& 0.99 & 0.95 & 1.22 & 1.37
	& 92 & 94 & 91 & 87 \\

	\midrule

	\textbf{MNAR}
	& \multicolumn{16}{c}{Outcome 1} \\
	\cmidrule(lr){2-17}
	& \multicolumn{4}{c}{PB (\%)}
	& \multicolumn{4}{c}{ESD}
	& \multicolumn{4}{c}{ESD / SE}
	& \multicolumn{4}{c}{CP (\%)} \\
	\cmidrule(lr){2-5}
	\cmidrule(lr){6-9}
	\cmidrule(lr){10-13}
	\cmidrule(lr){14-17}
	Method
	& $\beta_1$ & $\beta_2$ & $\beta_3$ & $\beta_4$
	& $\beta_1$ & $\beta_2$ & $\beta_3$ & $\beta_4$
	& $\beta_1$ & $\beta_2$ & $\beta_3$ & $\beta_4$
	& $\beta_1$ & $\beta_2$ & $\beta_3$ & $\beta_4$ \\
	\midrule

	SHIM
	& 6.6 & 5.9 & 10.5 & 20.5
	& 0.11 & 0.12 & 0.13 & 0.15
	& 0.93 & 1.00 & 1.00 & 1.10
	& 92 & 91 & 94 & 83 \\

	SHIM\_U
	& 6.3 & 6.4 & 10.3 & 21.3
	& 0.12 & 0.12 & 0.14 & 0.17
	& 0.97 & 1.01 & 1.07 & 1.19
	& 93 & 92 & 91 & 83 \\

	Lasso
	& 21.3 & 22.8 & 32.3 & 39.1
	& 0.16 & 0.14 & 0.15 & 0.18
	& 1.19 & 1.01 & 1.14 & 1.30
	& 55 & 54 & 67 & 56 \\

	Glasso
	& 8.7 & 8.5 & 13.0 & 4.8
	& 0.37 & 0.37 & 0.48 & 0.36
	& 0.93 & 1.00 & 1.28 & 0.96
	& 98 & 98 & 85 & 98 \\

	GIGG
	& 9.3 & 10.8 & 23.6 & 31.5
	& 0.17 & 0.13 & 0.22 & 0.23
	& 1.12 & 0.88 & 1.26 & 1.38
	& 89 & 89 & 83 & 76 \\

	MBSGSSS
	& 8.1 & 9.6 & 14.7 & 20.5
	& 0.14 & 0.12 & 0.22 & 0.22
	& 1.03 & 0.89 & 1.35 & 1.36
	& 88 & 89 & 86 & 88 \\

	\midrule

	&
	\multicolumn{16}{c}{Outcome 2} \\
	\cmidrule(lr){2-17}
	& \multicolumn{4}{c}{PB (\%)}
	& \multicolumn{4}{c}{ESD}
	& \multicolumn{4}{c}{ESD / SE}
	& \multicolumn{4}{c}{CP (\%)} \\
	\cmidrule(lr){2-5}
	\cmidrule(lr){6-9}
	\cmidrule(lr){10-13}
	\cmidrule(lr){14-17}
	Method
	& $\beta_1$ & $\beta_2$ & $\beta_3$ & $\beta_4$
	& $\beta_1$ & $\beta_2$ & $\beta_3$ & $\beta_4$
	& $\beta_1$ & $\beta_2$ & $\beta_3$ & $\beta_4$
	& $\beta_1$ & $\beta_2$ & $\beta_3$ & $\beta_4$ \\
	\midrule

	SHIM
	& 6.4 & 6.6 & 14.8 & 18.1
	& 0.12 & 0.12 & 0.14 & 0.13
	& 1.02 & 0.98 & 1.06 & 1.01
	& 90 & 91 & 91 & 89 \\

	SHIM\_U
	& 6.7 & 7.2 & 16.2 & 19.7
	& 0.13 & 0.12 & 0.16 & 0.16
	& 1.04 & 0.95 & 1.11 & 1.17
	& 88 & 91 & 89 & 86 \\

	Lasso
	& 20.2 & 22.5 & 33.8 & 36.5
	& 0.18 & 0.15 & 0.16 & 0.16
	& 1.27 & 1.12 & 1.17 & 1.21
	& 54 & 49 & 69 & 65 \\

	Glasso
	& 6.4 & 11.0 & 1.2 & 3.2
	& 0.37 & 0.41 & 0.43 & 0.41
	& 0.91 & 1.09 & 1.12 & 1.03
	& 92 & 92 & 90 & 98 \\

	GIGG
	& 8.8 & 11.1 & 24.2 & 27.9
	& 0.17 & 0.15 & 0.21 & 0.21
	& 1.10 & 0.98 & 1.22 & 1.21
	& 85 & 83 & 83 & 83 \\

	MBSGSSS
	& 7.7 & 10.0 & 15.8 & 18.5
	& 0.15 & 0.13 & 0.21 & 0.22
	& 1.05 & 0.98 & 1.31 & 1.36
	& 90 & 87 & 86 & 88 \\
	
	\midrule
	
	\textbf{ULTRA}
	& \multicolumn{16}{c}{Outcome 1} \\
	\cmidrule(lr){2-17}
	& \multicolumn{4}{c}{PB (\%)}
	& \multicolumn{4}{c}{ESD}
	& \multicolumn{4}{c}{ESD / SE}
	& \multicolumn{4}{c}{CP (\%)} \\
	\cmidrule(lr){2-5}
	\cmidrule(lr){6-9}
	\cmidrule(lr){10-13}
	\cmidrule(lr){14-17}
	Method
	& $\beta_1$ & $\beta_2$ & $\beta_3$ & $\beta_4$
	& $\beta_1$ & $\beta_2$ & $\beta_3$ & $\beta_4$
	& $\beta_1$ & $\beta_2$ & $\beta_3$ & $\beta_4$
	& $\beta_1$ & $\beta_2$ & $\beta_3$ & $\beta_4$ \\
	\midrule
	
	SHIM
	& 3.6 & 2.6 & 9.8 & 11.1
	& 0.14 & 0.14 & 0.18 & 0.15
	& 1.05 & 1.09 & 1.28 & 1.05
	& 91 & 89 & 89 & 91 \\
	
	SHIM\_U
	& 3.9 & 2.5 & 12.2 & 11.2
	& 0.14 & 0.13 & 0.19 & 0.15
	& 1.08 & 1.07 & 1.36 & 1.07
	& 92 & 90 & 86 & 93 \\
	
	Lasso
	& 16.2 & 16.7 & 34.7 & 29.0
	& 0.17 & 0.15 & 0.17 & 0.14
	& 1.12 & 1.13 & 1.25 & 1.03
	& 75 & 63 & 71 & 76 \\
	
	
	GIGG
	& 3.0 & 3.2 & 21.2 & 16.3
	& 0.16 & 0.16 & 0.25 & 0.23
	& 1.02 & 1.05 & 1.47 & 1.27
	& 92 & 90 & 81 & 87 \\
	
	MBSGSSS
	& 1.7 & 2.4 & 10.4 & 0.9
	& 0.15 & 0.14 & 0.21 & 0.15
	& 1.02 & 1.02 & 1.45 & 1.05
	& 92 & 93 & 88 & 95 \\
	
	\midrule
	
	&
	\multicolumn{16}{c}{Outcome 2} \\
	\cmidrule(lr){2-17}
	& \multicolumn{4}{c}{PB (\%)}
	& \multicolumn{4}{c}{ESD}
	& \multicolumn{4}{c}{ESD / SE}
	& \multicolumn{4}{c}{CP (\%)} \\
	\cmidrule(lr){2-5}
	\cmidrule(lr){6-9}
	\cmidrule(lr){10-13}
	\cmidrule(lr){14-17}
	Method
	& $\beta_1$ & $\beta_2$ & $\beta_3$ & $\beta_4$
	& $\beta_1$ & $\beta_2$ & $\beta_3$ & $\beta_4$
	& $\beta_1$ & $\beta_2$ & $\beta_3$ & $\beta_4$
	& $\beta_1$ & $\beta_2$ & $\beta_3$ & $\beta_4$ \\
	\midrule
	
	SHIM
	& 2.4 & 0.9 & 14.0 & 4.6
	& 0.12 & 0.14 & 0.17 & 0.14
	& 0.93 & 1.07 & 1.19 & 1.02
	& 96 & 94 & 88 & 95 \\
	
	SHIM\_U
	& 2.1 & 1.1 & 15.7 & 4.6
	& 0.12 & 0.13 & 0.18 & 0.14
	& 0.90 & 1.09 & 1.25 & 1.04
	& 96 & 91 & 87 & 96 \\
	
	Lasso
	& 17.7 & 16.6 & 36.9 & 24.4
	& 0.16 & 0.17 & 0.16 & 0.15
	& 1.10 & 1.22 & 1.15 & 1.06
	& 73 & 74 & 66 & 82 \\
	
	
	GIGG
	& 3.4 & 2.0 & 24.6 & 8.8
	& 0.17 & 0.16 & 0.24 & 0.22
	& 1.03 & 1.07 & 1.36 & 1.23
	& 93 & 91 & 81 & 91 \\
	
	MBSGSSS
	& 2.6 & 1.4 & 12.7 & 2.7
	& 0.14 & 0.13 & 0.20 & 0.15
	& 0.98 & 1.00 & 1.40 & 1.00
	& 93 & 94 & 90 & 95 \\
	
	\bottomrule
	\end{tabular}
\end{table}

\section{Neuroimaging and CSF Acquisition Details}
\label{app:data}

T1-weighted 
magnetization-prepared rapid gradient-echo (MPRAGE) images were 
acquired with a repetition time (TR) of $8.9~\mathrm{ms}$, an echo 
time (TE) of $4.6~\mathrm{ms}$, and an isotropic spatial resolution 
of $1 \times 1 \times 1~\mathrm{mm}^3$. The images were 
skull-stripped using SynthStrip \cite{hoopes2022synthstrip} and 
parcellated into anatomical regions of interest (ROIs) using 
NiChart\_DLMUSE~\cite{bashyam2025dlmuse}. Regional GM volumes were 
derived from the resulting anatomical parcellations.

Pseudo-continuous arterial spin-labeling (pCASL) label-control image 
pairs were acquired to quantify CBF. Acquisition parameters included 
a labeling duration of $1.65~\mathrm{s}$, a post-labeling delay of 
$1.525~\mathrm{s}$, a spatial resolution of $3 \times 3 \times 
7~\mathrm{mm}^3$, a TR of $4{,}000~\mathrm{ms}$, and a TE of 
$13~\mathrm{ms}$. Motion correction was performed using 
\texttt{3dvolreg} in AFNI (version 23.1.10). Regional CBF estimates 
were derived for the anatomical ROIs and corrected for partial-volume 
effects to reduce the influence of regional atrophy on the estimated 
perfusion values.

Cerebrospinal fluid collection was an optional component of the VMAP protocol. Participants undergoing CSF collection completed a morning fasting lumbar puncture using a Sprotte 25-gauge spinal needle. Samples were immediately mixed and centrifuged at $2{,}000~\mathrm{g}$ and $4^\circ\mathrm{C}$ for 10 minutes, after which the supernatant was aliquoted into polypropylene tubes and stored at $-80^\circ\mathrm{C}$. CSF $A\beta_{1-42}$ concentrations were measured in batch using the Fujirebio INNOTEST $\beta\text{-AMYLOID}_{(1\text{--}42)}$ enzyme-linked immunosorbent assay. The reported intra-assay coefficients of variation were below $10\%$.

\end{appendix}

\end{document}